\documentclass[11pt,twoside]{article}

\usepackage[T1]{fontenc}
\usepackage[utf8]{inputenc}
\usepackage{lmodern}
\usepackage[english]{babel}
\usepackage[a4paper,margin=1in,headsep=18pt]{geometry}
\usepackage{amsmath,amssymb,amsthm}
\usepackage{graphicx}
\usepackage{xcolor}
\usepackage[hidelinks]{hyperref}
\providecommand{\doi}[1]{\href{https://doi.org/#1}{doi:#1}}
\usepackage{caption}
\usepackage{subcaption}
\usepackage{microtype}
\usepackage{multirow}
\usepackage{float}
\usepackage{fancyhdr}

\theoremstyle{definition}
\newtheorem{definition}{Definition}[section]
\theoremstyle{plain}
\newtheorem{proposition}[definition]{Proposition}
\newtheorem{lemma}[definition]{Lemma}
\newtheorem{theorem}[definition]{Theorem}
\theoremstyle{remark}
\newtheorem*{remark}{Remark}

\title{Evaluating Neural Probabilistic Forecasts of
Sparse Earthquake Counts:
What Survives a Careful Evaluation}
\author{Igilik Alim\\[2pt]
\small International Information Technology University (IITU), Almaty, Kazakhstan\\
\small \href{mailto:al1mkamindcore@gmail.com}{\texttt{al1mkamindcore@gmail.com}}}
\date{}

\begin{document}

\maketitle
\thispagestyle{fancy}

\begin{abstract}
We revisit a neural negative binomial model for weekly earthquake counts on a $3^\circ \times 3^\circ$ grid over the Tian Shan and Pamir, and ask which of its conclusions survive a careful evaluation. At the completeness magnitude of~4.5, the 144-week test period holds 2448 cell-weeks; six cell-weeks, in five distinct weeks, have three or more events.

Two conclusions survive. First, Poisson forecasts, temporal ETAS included, understate the probability of such a busy week: ETAS expects 0.72 where six occurred, a shortfall by a factor of~8.3 ($p = 0.00009$, exact test). One dispersion parameter removes it (ETAS with dispersion expects 4.45, $p = 0.58$); simulating aftershock chains within the week also reproduces the tail but overshoots the total count by about a quarter. Second, the network's advantage over a linear model disappears once that model has an intercept per cell: the remaining gain is 3.4\,\% of deviance, not significant ($p = 0.57$ over five test years), and it comes from one cell-week out of the 4420 scored; without it the network is 0.7\,\% worse.

A blind test on 2024--September 2026, with every model fitted on data through 2023, confirms the tail result: Poisson forecasts are rejected ($p$ from $3 \cdot 10^{-8}$ to 0.001), dispersion models are not ($p$ between 0.25 and 0.84), except the NB GLM, whose level blew up out of sample and which is rejected at $p = 0.003$. Out of sample, model families differ in how they extrapolate beyond the training range, not in the per-cell term.

We also measure the evaluation traps met on the way: a biased reference level, a single seed (five of six effects in v1's results table change sign), too few bootstrap resamples, and a threshold below completeness, which shrinks the ETAS shortfall from 8.3 to 2.8 times. This version replaces preprint~v1; Section~\ref{sec:corrections} lists the corrections.
\end{abstract}

\begin{center}
\small
\textbf{Code availability:} The implementation, the training and
evaluation pipeline and the blind-test driver are publicly available at\\
\url{https://github.com/Al1mkaYandere/seismic-probabilistic-modeling}
\end{center}

\section{Introduction}

Earthquake forecasting is a critical task for natural risk management, infrastructure resilience planning, and emergency response operations. For Central Asia, and the Tian Shan mountain system in particular, this problem carries heightened importance due to high tectonic activity, complex geodynamics, and pronounced spatiotemporal heterogeneity of seismic processes. In the applied setting, the goal is not a deterministic forecast of individual events, but a macroscopic forecast of seismicity intensity: estimating the expected number of earthquakes above the catalogue's completeness magnitude ($M \ge 4.5$ here) on a spatial grid at a weekly horizon.

Historically, count data forecasting in fixed spatiotemporal cells has been formulated within the Poisson framework. However, its key assumption---equality of the conditional mean and conditional variance---is systematically violated in real seismological data. Earthquakes exhibit pronounced clustering associated with swarm activity, foreshock--aftershock sequences, and episodes of anomalous activity, resulting in overdispersion in which the variance substantially exceeds the mean. Under these conditions, uncritical application of the Poisson distribution leads to biased uncertainty estimates and, consequently, to underestimation of the risk of extreme scenarios.

Machine-learning forecasts of seismicity are also easy to evaluate wrongly. A model of counts scored with a continuous loss such as the mean squared error ignores that counts are discrete and that their spread matters, and a comparison against a weak baseline can make a model look better than it is~\NoHyper\cite{mignan2020neural}\endNoHyper.

The model examined here, introduced in preprint~v1~\NoHyper\cite{igilik2026v1}\endNoHyper, replaces the Poisson distribution by the Negative Binomial (NB), which allows for overdispersion. Its parameters come from a small neural network with two parts: (i) a learned embedding for each spatial cell, meant to absorb differences between cells; and (ii) a multilayer perceptron (MLP) on lagged features of past activity, including released seismic energy and the time since the last large event.

Neural models of this kind are now common in seismicity 
forecasting: neural point processes and catalogue encoders are 
compared against ETAS directly~\NoHyper\cite{stockman2023forecasting,zlydenko2023neural}\endNoHyper, 
and the practices for building and testing such forecasts have 
recently been surveyed~\NoHyper\cite{mizrahi2024review}\endNoHyper.

Sparse counts, however, make such forecasts hard to evaluate. Most cell-weeks are empty, the events that matter are rare, and a handful of observations can decide a comparison between models. General catalogues of evaluation errors exist for forecasting with machine learning~\NoHyper\cite{hewamalage2022forecast}\endNoHyper{} and for leakage in machine-learning-based science~\NoHyper\cite{kapoor2023leakage}\endNoHyper; in seismology, Mignan and Broccardo showed that simple models often match deep networks once the comparison is made fairly~\NoHyper\cite{mignan2019one,mignan2020neural}\endNoHyper.

This paper is a case study in that spirit, and it is a revision of our own preprint~v1. Its contribution is not a new model. It is a set of measured sizes: for one real, sparse catalogue, how large each evaluation trap is, and in which direction it moves the conclusion. Rechecking v1 showed that several of its claims rest on such traps (Section~\ref{sec:corrections}); here we report what remains once they are removed, and how large each one was.

The known theory behind the result that survives is that a branching process produces counts more dispersed than Poisson~\NoHyper\cite{kagan2010statistical}\endNoHyper. What we measure is how much this matters for the probability of a busy week in this catalogue, and whether one dispersion parameter is enough to capture it.

The contributions of this version are:
\begin{itemize}
    \item an exact test of the busy-week tail, applied to every model on the same test period: Poisson forecasts, ETAS included, are rejected, and forecasts carrying one dispersion parameter are not;

    \item a blind test on a period downloaded after the analysis was frozen (2024--September 2026), reported together with the three predictions, recorded beforehand, that did not come true;

    \item measured sizes for the evaluation traps this case study ran into, and an explicit list of the statements of preprint~v1 that are withdrawn or corrected.
\end{itemize}

\section{Data and Feature Engineering}

\subsection{Earthquake Catalog and Study Area}
\label{sec:catalog}

The empirical basis of this study is the open-access USGS (United States Geological Survey) catalog, covering 2 January 2010 to 27 December 2023. The analysis is restricted to the seismically active zone of Central Asia, encompassing the Tian Shan and Pamir mountain systems, within the geographic window $38^\circ$--$45^\circ$N and $65^\circ$--$85^\circ$E.

To ensure comparability of statistical estimates and reduce the influence of observational network limitations, the catalog was filtered by the magnitude of completeness~$M_c$. The threshold~$M_c$ was estimated using the Maximum Curvature method~\cite{wiemer2000minimum}: a frequency histogram of events in bins of $\Delta m = 0.1$ was constructed for the entire region, with the maximum occurring at $M \approx 4.3$; applying the correction of $+0.2$~\cite{woessner2005assessing} yields $M_c = 4.5$, with a $b$-value of $1.42$ (Aki--Utsu maximum-likelihood estimate~\cite{aki1965maximum} with
the half-bin correction, $n = 722$ events).

Preprint~v1 nevertheless used every event with $M \ge 3.0$, that is, a threshold below the completeness magnitude the same work had estimated. The catalogue is not complete at that level: 885 of the 1607 events above 3.0, or 55.1\,\%, lie below 4.5. This version therefore uses $M \ge 4.5$ as the main threshold and reports 3.0 only for comparison with~v1. The choice is not cosmetic: working below the completeness threshold hides most of the effect this paper is about, as Section~\ref{sec:tail} shows.

At threshold 4.5 the panel holds 711 events against 1589 at 3.0 (the raw catalogue holds 722 against 1607; the difference between panel and catalogue is the first 12 weeks of each cell, which are consumed by the lagged features). At threshold 4.5 the weekly
axis is pinned to that of the 3.0 catalogue (28 December 2009 to 25
December 2023); without the pin the last week, which holds no event above
4.5, would drop out and the test period would start and end one week
earlier. The main numbers of Sections~\ref{sec:tail} and~\ref{sec:network-vs-baselines}
are given at both thresholds; Table~\ref{tab:tail-main}, the static-protocol
score comparisons and the blind test of Section~\ref{sec:blind} are at 4.5
only.

\subsection{Spatiotemporal Grid and Target Distribution}

The study region is discretized into a regular spatial grid with 
a resolution of $3.0^\circ \times 3.0^\circ$. The temporal 
aggregation interval is set to one calendar week. For each spatial
cell $(i,j)$ and week $t$, the target variable is defined as
\[
Y^t_{i,j} := \text{number of recorded earthquakes in cell } 
(i,j) \text{ during week } t,
\]
where $Y^t_{i,j} \in \mathbb{N}_0$.

Empirical analysis of the target variable distribution reveals 
pronounced skewness: the overwhelming majority of observations 
correspond to zero counts (no events), while the right tail reaches
22 events in one cell-week at threshold 4.5 (35 at 3.0). This 
structure argues against fitting the counts with a squared-error loss.

\subsection{Feature Engineering}

The predictors are designed to reflect the physics of the seismic 
cycle:
\begin{itemize}
    \item \textbf{Seismic Energy.} Magnitude-to-energy transformation 
    based on the relation $E \propto 10^{1.5M}$.
    
    \item \textbf{Seismic Gap.} Time in weeks since the last event 
    with magnitude $M \ge 4.5$, modeling the local stress 
    accumulation phase.
    
    \item \textbf{Accumulated Activity.} Rolling windows of 4, 8 and
    12 weeks (largest magnitude, released energy and event count) to
    account for clustering and aftershock sequences.
\end{itemize}

All continuous predictors are standardized prior to training, while 
the target variable is retained in discrete form. Differences in the
average activity of cells are captured by a learned embedding of each
cell (Definition~\ref{def:architecture}); dependence between
neighbouring cells is not modelled by any model in this paper
(Section~\ref{sec:limitations}).


\section{Methodology}
\label{sec:methodology}
This section defines the data, the features and the two network models.
The probability space behind them, the derivation of the training
objective and the proofs of the results stated here are collected in
Appendix~\ref{app:formal}; none of them is needed to follow the
evaluation in Section~\ref{sec:experiments}.

\subsection{Data Formalization and Feature Space Construction}

\subsubsection{Formalization of the Source Catalog}

From the USGS catalog metadata, four components are retained for 
two-dimensional spatiotemporal modeling: coordinates, time, and 
magnitude.

\begin{definition}[Seismic Catalog]\label{def:catalog}
The filtered catalog is defined as a set of tuples
\[
\mathcal{D}_{raw} = \{(x_k, y_k, t_k, m_k)\},
\]
where $x_k, y_k, t_k \in \mathbb{R}$, with $x_k$ denoting 
longitude, $y_k$ latitude, and $m_k \ge M_c$ magnitude ($M_c = 4.5$ in the main analysis, 3.0
for comparison with v1). The index 
set
\[
\mathcal{I} = \{1,\dots,N\}
\]
induces the working set
\[
\widetilde{\mathcal{D}}_{raw} \subset \mathcal{I} \times 
\mathcal{D}_{raw} \implies \widetilde{\mathcal{D}}_{raw} = 
\bigcup_{k=1}^{N} \{(k, x_k, y_k, t_k, m_k)\},
\]
providing unique identification of each event by key $k$.
\end{definition}

\subsubsection{Spatiotemporal Discretization}

\begin{definition}[Spatial Domain and Discretization Operator]
\label{def:grid}
The study region
\[
\mathcal{S} = [x_{\min}, x_{\max}] \times [y_{\min}, y_{\max}]
\]
is representable as a disjoint union of cells $\mathcal{S}_{i,j}$ 
with steps $\Delta x$, $\Delta y$:
\[
\mathcal{S} = \bigsqcup_{i=0}^{n_x-1} \bigsqcup_{j=0}^{n_y-1} 
\mathcal{S}_{i,j},
\]
where $n_x$, $n_y$ denote the number of steps along the $x$ and 
$y$ axes. The discretization operator $H$ maps continuous 
coordinates to discrete indices:
\[
H(k, x_k, y_k, t_k, m_k) = (k, i_k, j_k, \tau_k, m_k),
\]
where
\[
i_k = \left\lfloor \frac{x_k - x_{\min}}{\Delta x} \right\rfloor, 
\quad
j_k = \left\lfloor \frac{y_k - y_{\min}}{\Delta y} \right\rfloor, 
\quad
\tau_k = \left\lfloor \frac{t_k - t_{\min}}{\Delta t} \right\rfloor.
\]
(\,Temporal discretization is implemented as follows: each event,
with its time in UTC, is assigned to the week that contains it, running
from Monday to Sunday ($\tau = \texttt{pd.Period}(t_k, \texttt{"W"})$,
whose start is a Monday); the full regular time series for each cell is
built on the same Mondays (\texttt{pd.date\_range(..., freq="W-MON")}),
so event weeks and grid weeks coincide exactly.\,) The discrete catalog is then:
\[
\mathcal{D}_{grid} = \bigcup_{v \in \widetilde{\mathcal{D}}_{raw}} 
\{ H(v) \}.
\]
\end{definition}

\subsubsection{Construction of the Active Grid}

The discrete indices belong to finite sets
\[
\mathcal{I}_x = \{0, 1, \dots, n_x - 1\}, \quad
\mathcal{I}_y = \{0, 1, \dots, n_y - 1\}, \quad
\mathcal{T}_{all} = \{0, 1, \dots, \tau_{\max}\},
\]
where
\[
\tau_{\max} = \left\lfloor \frac{t_{\max} - t_{\min}}{\Delta t} 
\right\rfloor.
\]
The full spatiotemporal grid is:
\[
\Omega = \mathcal{I}_x \times \mathcal{I}_y \times \mathcal{T}_{all}.
\]

\begin{definition}[Index Set (Bucket)]\label{def:bucket}
For $(i,j,\tau) \in \Omega$:
\[
\mathcal{B}_{i,j}^{(\tau)} = \big\{ k \in \{1, \dots, N\} \bigm| 
H(k, x_k, y_k, t_k, m_k) = (k, i, j, \tau, m_k) \big\}
\]
\[
= \big\{ k \in \{1, \dots, N\} \bigm| (k, i, j, \tau, m_k) \in 
\mathcal{D}_{grid} \big\}.
\]
\end{definition}

\begin{definition}[Active Grid]\label{def:active-grid}
The set of active spatial indices:
\[
\mathcal{S}_{active} = \{ (i,j) \in \mathcal{I}_x \times 
\mathcal{I}_y \mid \exists\, \tau \in \mathcal{T}_{all} \colon 
\mathcal{B}_{i,j}^{(\tau)} \neq \emptyset \}.
\]
The working grid:
\[
\Omega_{active} = \mathcal{S}_{active} \times \mathcal{T}_{all}.
\]
\end{definition}

\begin{remark}
The restriction to $\mathcal{S}_{active}$ excludes aseismic nodes 
while preserving time series continuity for active nodes (including 
zero observations), which is necessary for correct construction of 
lagged features.
\end{remark}

\subsubsection{Aggregation and Feature Vector}

\begin{definition}[Aggregated Statistics]\label{def:aggregation}
For $(i,j,\tau)\in\Omega_{active}$, the target variable is:
\[
Y_{i,j}^{(\tau)} = \big|\mathcal{B}_{i,j}^{(\tau)}\big|.
\]
Total released seismic energy:
\[
E_{i,j}^{(\tau)} = \sum_{k \in \mathcal{B}_{i,j}^{(\tau)}} 
10^{1.5 m_k}.
\]
Extreme magnitudes:
\[
M_{\max, i,j}^{(\tau)} = \max_{k \in \mathcal{B}_{i,j}^{(\tau)}} 
\{ m_k \}, \quad
M_{\min, i,j}^{(\tau)} = \min_{k \in \mathcal{B}_{i,j}^{(\tau)}} 
\{ m_k \}.
\]
When $\mathcal{B}_{i,j}^{(\tau)} = \emptyset$, we set
\[
Y_{i,j}^{(\tau)} = 0, \quad E_{i,j}^{(\tau)} = 0,
\]
\[
M_{\max, i,j}^{(\tau)} = 0, \quad M_{\min, i,j}^{(\tau)} = 0.
\]
\end{definition}

The aggregated dataset is:
\[
\mathcal{D}_{agg} = \bigcup_{(i,j,\tau) \in \Omega_{active}}
\Big\{ \big( i, j, \tau, Y_{i,j}^{(\tau)}, E_{i,j}^{(\tau)}, 
M_{\max, i,j}^{(\tau)}, M_{\min, i,j}^{(\tau)} \big) \Big\}.
\]

\begin{definition}[Feature Functional]\label{def:features}
Let $W_{\max}=12$ and $\mathcal{T}_{target} = \{ t \in 
\mathcal{T}_{all} \mid t \ge W_{\max} \}$. Define
\[
\Phi: \mathcal{S}_{active} \times \mathcal{T}_{target} \to 
\mathbb{R}^7,
\]
where
\[
\Phi(i,j,t)=\big(\phi_1(i,j,t),\dots,\phi_7(i,j,t)\big)^T,
\quad (i,j)\in\mathcal{S}_{active},\; t\in\mathcal{T}_{target},
\]
with components
\[
\phi_1(i,j,t)=Y_{i,j}^{(t-1)},
\]
\[
\phi_2(i,j,t)=M_{\max,i,j}^{(t-1)},
\]
\[
\phi_3(i,j,t)=M_{\min,i,j}^{(t-1)},
\]
\[
\phi_4(i,j,t)=\max_{\tau\in\{t-4,\dots,t-1\}} 
M_{\max,i,j}^{(\tau)},
\]
\[
\phi_5(i,j,t)=\sum_{\tau=t-12}^{t-1}Y_{i,j}^{(\tau)},
\]
\[
\phi_6(i,j,t)=\sum_{\tau=t-8}^{t-1}E_{i,j}^{(\tau)},
\]
\[
\phi_7(i,j,t)=(t-1)-\max\{\tau\le t-1\mid 
M_{\max,i,j}^{(\tau)}\ge 4.5\}.
\]
If the index set in $\phi_7$ is empty, the value $500.0$ is 
assigned.
\end{definition}

Let $\mathbf{X}_{i,j}^{(t)} := \Phi(i,j,t)$. The final dataset is:
\[
\mathcal{D}_{final}=\bigcup_{(i,j)\in\mathcal{S}_{active}}\;
\bigcup_{t\in\mathcal{T}_{target}}\Big\{\big(\mathbf{X}_{i,j}^{(t)}
,\;Y_{i,j}^{(t)}\big)\Big\}.
\]

\subsubsection{Operator Pipeline for Data Transformation}

\[
\begin{aligned}
\mathcal{D}_{raw}
&\xrightarrow{\text{Index }(\mathcal{I} \times \cdot)}
\widetilde{\mathcal{D}}_{raw}
\xrightarrow{H}
\mathcal{D}_{grid}
\xrightarrow{\text{Bucketing}}
\{\mathcal{B}_{i,j}^{(\tau)}\}
\\[2pt]
&\xrightarrow{\mathcal{A}\;\text{(aggregation on }\Omega_{active})}
\mathcal{D}_{agg}
\xrightarrow{\Phi\;\text{(lag/rolling features)}}
\mathcal{D}_{final}.
\end{aligned}
\]

The composition of mappings transforms a continuous spatiotemporal 
point process into a tensor representation suitable for neural 
network architectures with shared weights. The hierarchy of temporal 
windows (4, 8, 12 weeks) implements multi-scale memory of the 
process. The feature vector $\mathbf{X}_{i,j}^{(t)}$ is constructed 
exclusively from $\tau \le t-1$, thereby precluding any look-ahead 
bias.

\subsection{Empirical Risk Minimization}

For the training index set $\Omega_{train} \subset 
\mathcal{S}_{active} \times \mathcal{T}_{all}$, define
\[
\mathbf{Y}_{train} = \{ Y_{i,j}^{(t)} \mid (i,j,t) \in 
\Omega_{train} \}, \quad \mathbf{X}_{train} = 
\{ \mathbf{X}_{i,j}^{(t)} \mid (i,j,t) \in \Omega_{train} \}.
\]

\begin{definition}[Training Index Set for Static 
Evaluation]\label{def:omega-train}
Let the week set $\mathcal{T}_{all}$ be chronologically ordered 
and partitioned into disjoint blocks
\[
\mathcal{T}_{train} \cup \mathcal{T}_{test} = \mathcal{T}_{all},
\qquad
\mathcal{T}_{train} \cap \mathcal{T}_{test} = \varnothing,
\qquad
|\mathcal{T}_{train}| = \lfloor 0.8\,|\mathcal{T}_{all}|\rfloor,
\]
where $\mathcal{T}_{train}$ is the initial chronological prefix 
and $\mathcal{T}_{test}$ is the remaining suffix. The training 
and test index sets are then defined as
\[
\Omega_{train}
=\bigl\{(i,j,t)\in \mathcal{S}_{active}\times\mathcal{T}_{all}
\mid t\in\mathcal{T}_{train}\bigr\}
\cap \operatorname{supp}(\mathcal{D}_{final}),
\]
\[
\Omega_{test}
=\bigl\{(i,j,t)\in \mathcal{S}_{active}\times\mathcal{T}_{all}
\mid t\in\mathcal{T}_{test}\bigr\}
\cap \operatorname{supp}(\mathcal{D}_{final}).
\]
\end{definition}

Thus, $\Omega_{train}$ is not an arbitrary subset of 
$\mathcal{S}_{active} \times \mathcal{T}_{all}$: it is a 
synchronized calendar block containing all active cells in the 
early weeks present in $\mathcal{D}_{final}$. Test weeks are 
entirely excluded from training, yielding an out-of-time 
evaluation under the static 80/20 split.

Invoking the filtration $\mathcal{F}_{t-1}$ and the sufficiency 
of $\Phi$ (Appendix~\ref{app:formal}), we postulate spatiotemporal conditional independence 
on the training block:
\[
Y_{i,j}^{(t)} \perp Y_{k,m}^{(\tau)}
\mid \{\mathbf{X}_{i',j'}^{(t')}\}_{(i',j',t')\in\Omega_{train}}
\quad
\forall\, (i,j,t) \neq (k,m,\tau) \in \Omega_{train},
\]
which yields the factorization of the joint conditional 
likelihood.

\begin{remark}[Spatial Conditional Independence 
Assumption]\label{rem:spatial-independence}
This assumption is standard in GLM and neural count regression 
models, but may be violated in seismology: Coulomb stress 
transfer and ETAS triggering~\cite{ogata1988statistical} induce 
cross-cell dependencies. The factorization below is a marginal 
predictive factorization, not a full generative model of the 
joint spatiotemporal field. Empirical verification via Moran's 
$I$ on standardized Pearson residuals was carried out in 
preprint~v1; its significance there rested on a single cell 
(Section~\ref{sec:corrections}). Measured dependence between 
cells is small but not zero (Section~\ref{sec:limitations}), so 
the proposed model is to be read as a \emph{marginal 
predictor}. Extension to an ETAS-like 
spatial convolution is left for future work 
(see Section~\ref{sec:limitations}).
\end{remark}

\[
\mathbb{P}(\mathbf{Y}_{train} \mid \mathbf{X}_{train}) = 
\prod_{(i,j,t) \in \Omega_{train}} 
\mathbb{P}(Y_{i,j}^{(t)} \mid \mathbf{X}_{i,j}^{(t)}).
\]
The parametric approximation $\mathbb{P}_{\theta}$ inherits 
this factorization structure.

\begin{remark}[Relationship to Year-by-Year Validation]\label{rem:wf-vs-static}
The static 80/20 split defined above differs structurally from the year-by-year (walk-forward) protocol of Section~\ref{sec:eval-protocol}. In the static split, $\Omega_{train}$ is fixed once; in the year-by-year protocol, a separate $\Omega_{train}^{(Y)}$ is constructed for each test year $Y \in \{2018, \dots, 2023\}$, expanding as $Y$ increases. The static split carries the tail evaluation of Section~\ref{sec:tail} and the score comparisons of Section~\ref{sec:network-vs-baselines}; the year-by-year protocol carries that section's deviance comparisons. Numbers are never compared across the two, because the two train the network with different settings (Section~\ref{sec:eval-protocol}).
\end{remark}

The network is trained by minimising the negative log-likelihood of the
training block,
\[
\mathcal{L}_{NLL}(\theta) = -\sum_{(i,j,t) \in \Omega_{train}}
\log \mathbb{P}_{\theta}(Y_{i,j}^{(t)} \mid
\mathbf{X}_{i,j}^{(t)}),
\]
which, under the factorization above, is the empirical counterpart of
minimising the expected Kullback--Leibler divergence from the true
conditional distribution (Proposition~\ref{prop:nll} in
Appendix~\ref{app:formal}).

\subsubsection{Baseline Poisson Approximation}

\begin{definition}[Poisson Parametrization]\label{def:poisson}
The neural network $f_{NN}(\cdot;\theta): \mathbb{R}^7 \to 
\mathbb{R}$ induces the model
\[
Y_{i,j}^{(t)} \mid \mathbf{X}_{i,j}^{(t)} \sim 
\mathrm{Poisson}(\lambda_{i,j}^{(t)}),
\]
with the intensity given by the same softplus output as in
Definition~\ref{def:architecture}:
\[
\lambda_{i,j}^{(t)} = \ln\!\big(1 + \exp\big(f_{NN}
(\mathbf{X}_{i,j}^{(t)};\theta)\big)\big) + \epsilon > 0.
\]
The softplus output guarantees $\lambda_{i,j}^{(t)} > 0$ for
all inputs without imposing explicit constraints on the network
output. The Poisson network is the network of
Definition~\ref{def:architecture} trained with the Poisson loss below;
its second output is not used. The conditional probability mass function is:
\[
\mathbb{P}_{\theta}(Y_{i,j}^{(t)} = y \mid \mathbf{X}_{i,j}^{(t)}) 
= \frac{(\lambda_{i,j}^{(t)})^y\, 
e^{-\lambda_{i,j}^{(t)}}}{y!}, \qquad y \in \mathbb{N}_0.
\]
Taking the logarithm:
\[
\log \mathbb{P}_{\theta}(Y_{i,j}^{(t)} = y \mid 
\mathbf{X}_{i,j}^{(t)}) 
= y \log \lambda_{i,j}^{(t)} 
- \lambda_{i,j}^{(t)} - \log(y!).
\]
The term $\log(y!)$ does not depend on $\theta$. Dropping it 
and negating to obtain a minimization objective, the per-sample 
contribution to the NLL becomes $\lambda_{i,j}^{(t)} - 
Y_{i,j}^{(t)}\log\lambda_{i,j}^{(t)}$. Summing over 
$\Omega_{train}$ yields the Poisson loss:
\[
\mathcal{L}_{Poisson}(\theta) 
= \sum_{(i,j,t) \in \Omega_{train}} 
\!\Big( \lambda_{i,j}^{(t)} 
- Y_{i,j}^{(t)} \log \lambda_{i,j}^{(t)} \Big).
\]
\end{definition}

The defining property of the Poisson distribution is 
equidispersion: the conditional mean and conditional variance 
coincide:
\[
\mathbb{E}[Y_{i,j}^{(t)} \mid \mathbf{X}_{i,j}^{(t)}] 
= \mathrm{Var}[Y_{i,j}^{(t)} \mid \mathbf{X}_{i,j}^{(t)}] 
= \lambda_{i,j}^{(t)}.
\]

\begin{theorem}[Overdispersion from Latent 
Heterogeneity]\label{thm:overdispersion}
Suppose there exists a $\sigma$-algebra $\mathcal{F}^*$ with 
$\sigma(\mathbf{X}) \subset \mathcal{F}^*$, and an 
$\mathcal{F}^*$-measurable random variable $\lambda^* > 0$ with 
$\mathbb{E}[(\lambda^*)^2] < \infty$ such 
that $Y \mid \mathcal{F}^* \sim \mathrm{Poisson}(\lambda^*)$. 
If $\lambda^*$ is not $\sigma(\mathbf{X})$-measurable, then
$\mathrm{Var}[Y \mid \mathbf{X}] \ge \mathbb{E}[Y \mid \mathbf{X}]$ almost surely, 
with strict inequality on an event of positive probability.
\end{theorem}

The proof is given in Appendix~\ref{app:proofs}.

\begin{remark}
Theorem~\ref{thm:overdispersion} shows that overdispersion is 
a structural consequence of incomplete observable information, 
not an empirical artifact. In the seismological context, 
$\lambda^*$ represents the unobservable local seismogenic 
potential --- fluctuations in fault segment activation, 
aftershock cascades, and fluid migration --- which remain 
latent even after conditioning on the feature vector 
$\mathbf{X}$. The Negative Binomial distribution, introduced 
in Section~\ref{sec:nb}, provides an analytically tractable 
marginal model for $Y$ under a Gamma prior on $\lambda^*$.
\end{remark}

The empirical diagnostics associated with 
Theorem~\ref{thm:overdispersion} are reported in 
Figure~\ref{fig:overdispersion-diagnostics}. We use three 
explicit graphical designations: 
Figure~\ref{fig:target-count-tail-diagnostic} for the marginal 
tail of the count target, 
Figure~\ref{fig:local-dispersion-index-diagnostic} for the 
local dispersion index, and 
Figure~\ref{fig:mean-variance-diagnostic} for the 
mean--variance relationship.

\begin{figure}[!htbp]
    \centering
    \begin{subfigure}[t]{0.47\linewidth}
        \centering
        \includegraphics[width=\linewidth]{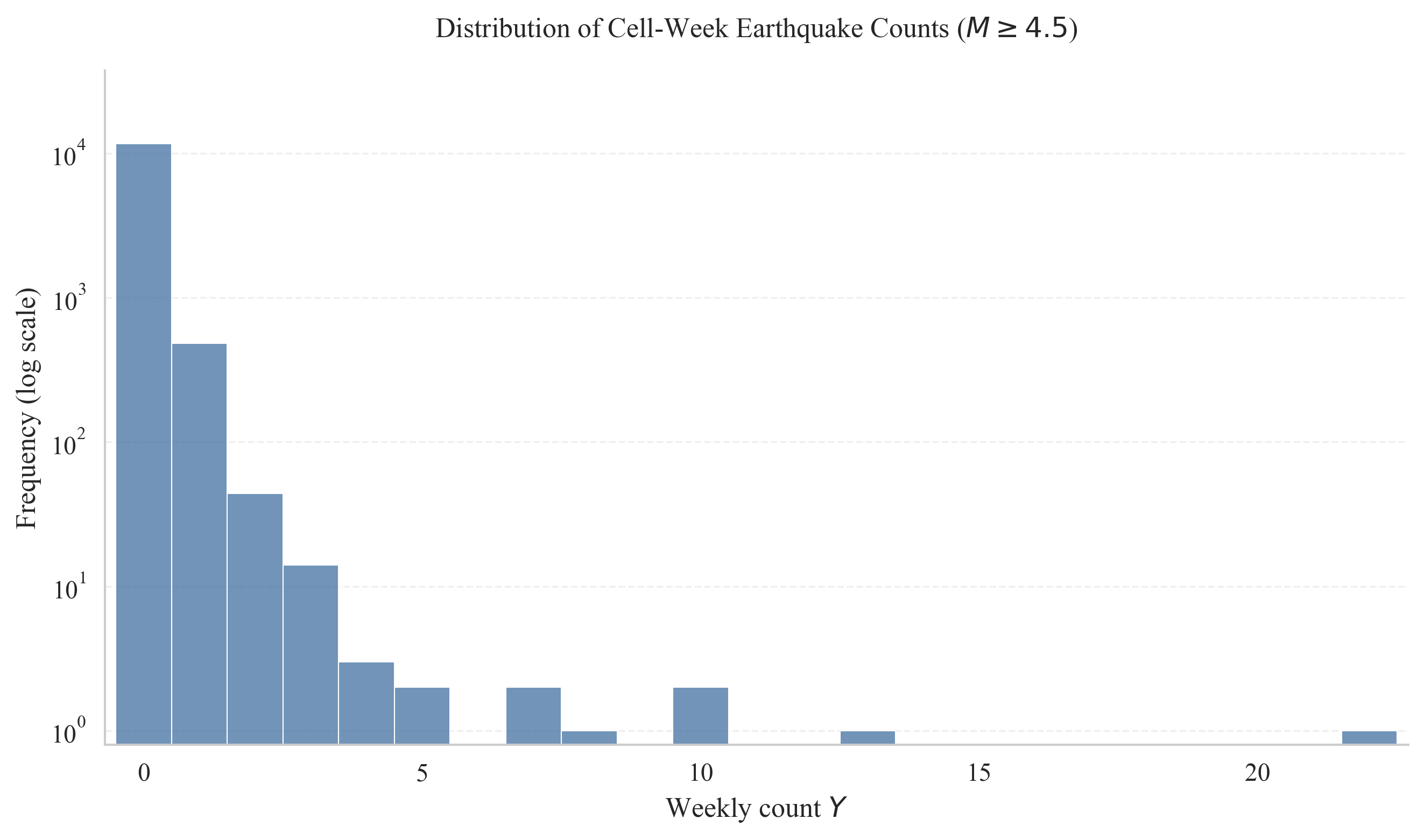}
        \caption{Target-count tail diagnostic.}
        \label{fig:target-count-tail-diagnostic}
    \end{subfigure}\hfill
    \begin{subfigure}[t]{0.47\linewidth}
        \centering
        \includegraphics[width=\linewidth]{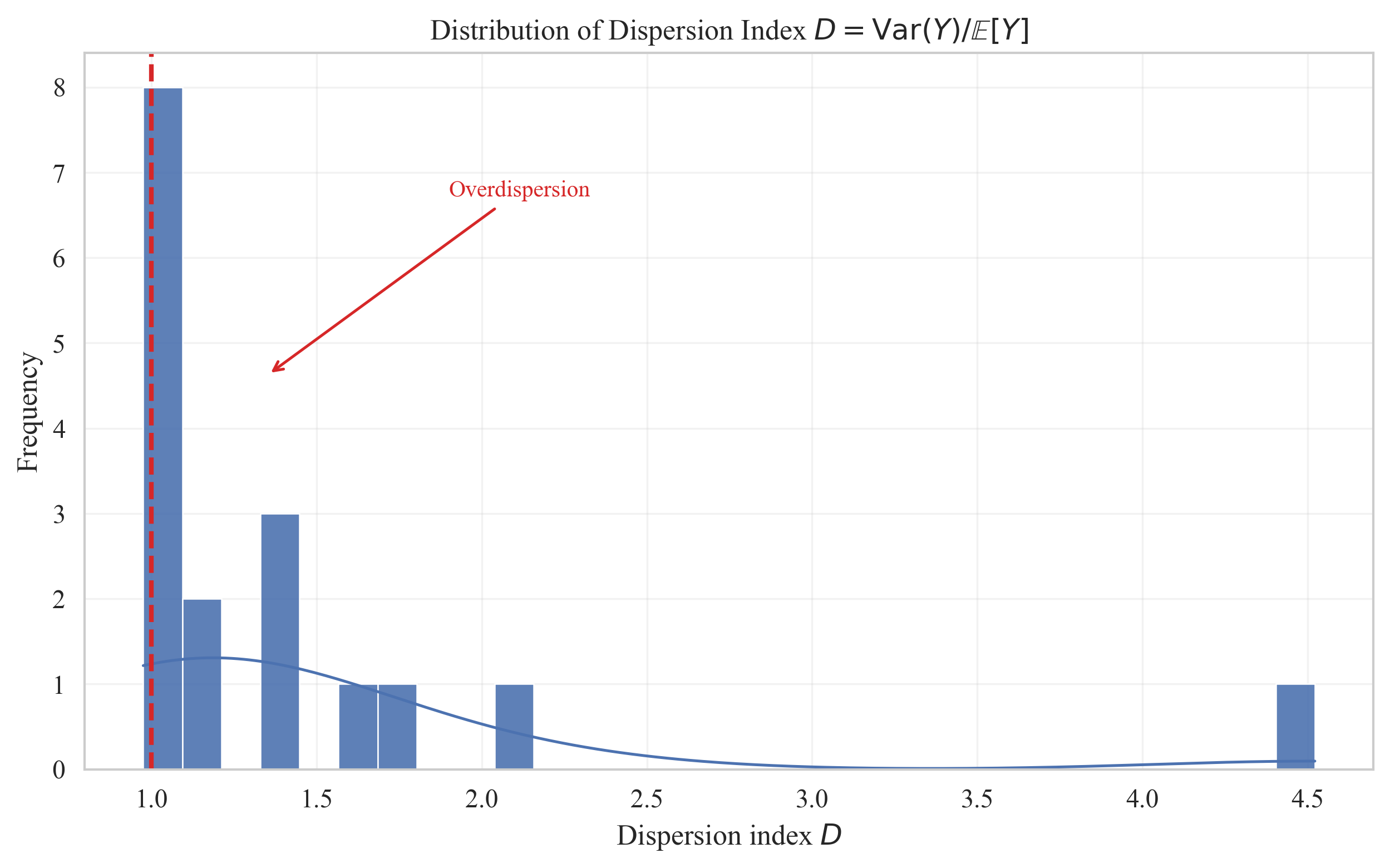}
        \caption{Local dispersion index diagnostic.}
        \label{fig:local-dispersion-index-diagnostic}
    \end{subfigure}

    \vspace{0.6em}
    \begin{subfigure}[t]{0.56\linewidth}
        \centering
        \includegraphics[width=\linewidth]{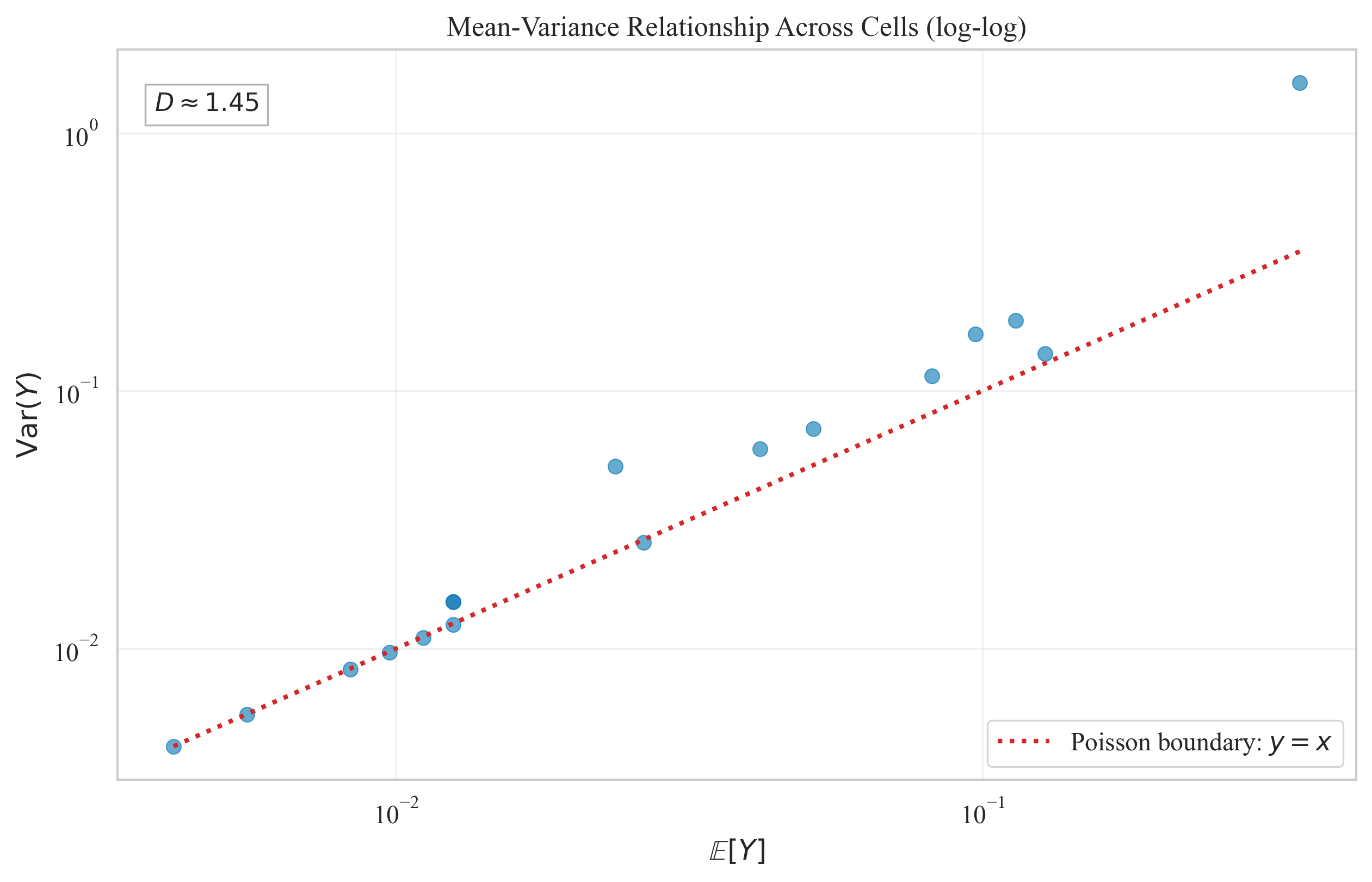}
        \caption{Mean--variance diagnostic.}
        \label{fig:mean-variance-diagnostic}
    \end{subfigure}
    \caption{Overdispersion diagnostics following
    Theorem~\ref{thm:overdispersion}, at threshold $M_c = 4.5$ on the
    whole panel (17 cells, 719 weeks): the marginal distribution of $Y$,
    the local dispersion index $D=\mathrm{Var}(Y)/\mathbb{E}(Y)$ of each
    cell, and each cell's mean and variance against the Poisson relation
    $\mathrm{Var}=\mathbb{E}$. The label $D \approx 1.45$ in panel (c) is the
    mean of the 17 cells' indices; in panel (b) the dashed line marks $D = 1$
    and the arrow points to the overdispersed side.}
    \label{fig:overdispersion-diagnostics}
\end{figure}

The target-count tail diagnostic in 
Figure~\ref{fig:target-count-tail-diagnostic} shows a heavy 
right tail incompatible with rapid Poisson decay. The local
dispersion index in Figure~\ref{fig:local-dispersion-index-diagnostic}
exceeds 1 in 10 of the 17 cells (median 1.21, largest 4.52), and the
mean--variance diagnostic in Figure~\ref{fig:mean-variance-diagnostic}
places the cells with more events above the Poisson line
$\mathrm{Var}=\mathbb{E}$, while most of the sparsest cells lie on it.

Consequently, $\mathcal{L}_{Poisson}$ systematically 
under-specifies the model: equidispersion understates the 
probability of extreme events. A transition to distribution 
families with separate mean and dispersion parametrization is 
therefore required.

\subsubsection{Gamma--Poisson Mixture and the Negative 
Binomial Distribution}\label{sec:nb}

In the seismological context, the latent intensity $\lambda^*$ 
reflects unobservable variations in local seismogenic potential 
--- activation of fault segments, aftershock sequences, and 
swarm activity. The Gamma distribution for $\lambda^*$ is the 
conjugate prior to the Poisson observation model and admits 
an analytic marginalization.

\begin{proposition}[NB as the Marginal Distribution of the 
Gamma--Poisson Mixture]\label{prop:nb-mixture}
Let the random effect $\lambda^* \sim \mathrm{Gamma}(r, \beta)$, 
with shape $r > 0$ and rate $\beta > 0$, and let 
$Y \mid \lambda^* \sim \mathrm{Poisson}(\lambda^*)$. Then the 
marginal distribution of $Y$ is Negative Binomial:
\[
\mathbb{P}(Y = y) = \frac{\Gamma(y+r)}{\Gamma(y+1)\,\Gamma(r)}
(1-p)^r p^y, \quad y \in \mathbb{N}_0, \quad 
p = (1+\beta)^{-1}.
\]
\end{proposition}

The proof is given in Appendix~\ref{app:proofs}.

\begin{remark}[Physical Interpretation]
The continuous Gamma mixture of Poisson processes provides 
a natural mechanism for the overdispersion observed in 
seismic data: even when the macroscopic feature vector 
$\mathbf{X}$ is fixed, the true local intensity fluctuates 
due to unobservable geophysical processes --- Coulomb stress 
transfer, fluid migration, and cascading activation of fault 
sub-segments. The shape parameter $r$ governs the degree of 
this latent heterogeneity: as $r \to \infty$ (equivalently 
$\alpha = r^{-1} \to 0$), the Gamma mixing distribution 
concentrates around its mean and the NB distribution converges 
to Poisson, recovering the equidispersion limit.
\end{remark}

\subsubsection{Reparametrization for Deep Learning}

\begin{proposition}[Moments of NB and $(\mu,\alpha)$-Parametrization]\label{prop:nb-reparam}
For $\mathrm{NB}(r,\beta)$ with $\alpha := r^{-1}$ and 
$\mu := r/\beta$, the variance satisfies
\[
\mathrm{Var}[Y] = \mu + \alpha\mu^2.
\]
The probability mass function in the $(\mu,\alpha)$ 
coordinates is:
\[
\mathbb{P}(Y = y \mid \mu, \alpha) = 
\frac{\Gamma(y+\alpha^{-1})}{\Gamma(y+1)\,\Gamma(\alpha^{-1})}
\left(\frac{1}{1+\alpha\mu}\right)^{\alpha^{-1}}
\left(\frac{\alpha\mu}{1+\alpha\mu}\right)^y.
\]
\end{proposition}

The proof is given in Appendix~\ref{app:proofs}.

\begin{remark}
For $\alpha > 0$, the variance scales quadratically with the 
mean; as $\alpha \to 0$, the Gamma mixing distribution 
concentrates and the model converges asymptotically to the 
Poisson process, recovering equidispersion 
$\mathrm{Var}[Y] \to \mu$.
\end{remark}

\begin{definition}[Neural NB Model]\label{def:neural-nb}
The neural network $f_{NN}(\mathbf{X};\theta)$ outputs the 
local distribution parameters:
\[
\begin{bmatrix} \mu_{\theta}(\mathbf{X}) \\ 
\alpha_{\theta}(\mathbf{X}) \end{bmatrix} 
= f_{NN}(\mathbf{X};\theta), \qquad 
\mu_{\theta} > 0,\quad \alpha_{\theta} > 0.
\]
The induced conditional measure is:
\[
\mathbb{P}_{\theta}(Y = y \mid \mathbf{X}) 
:= \mathbb{P}\Big(Y = y \;\Big|\; 
\mu = \mu_{\theta}(\mathbf{X}),\;
\alpha = \alpha_{\theta}(\mathbf{X})\Big).
\]
The NLL training objective is:
\[
\mathcal{L}_{NB}(\theta) = -\sum_{(i,j,t) \in \Omega_{train}} 
\log \mathbb{P}_{\theta}\!\left(Y_{i,j}^{(t)} \mid 
\mathbf{X}_{i,j}^{(t)}\right).
\]
\end{definition}

Each spatial cell $(i,j) \in \mathcal{S}_{active}$ is 
re-indexed via the bijection
\[
g: \mathcal{S}_{active} \to \mathcal{C}, \qquad 
\mathcal{C} = \{1,\ldots,N_c\},
\]
setting $c = g(i,j)$ and defining
\[
Y_c^{(t)} := Y_{i,j}^{(t)}, \qquad 
\mathbf{X}_c^{(t)} := \mathbf{X}_{i,j}^{(t)}.
\]
The empirical risk then takes the form:
\[
\mathcal{L}_{NB}(\theta) = 
-\sum_{(c,t)\in\Omega_{train}} 
\ln \mathbb{P}_{\theta}\!\big(Y_c^{(t)}\mid 
\mathbf{X}_c^{(t)}\big).
\]

The architecture $f_{NN}$ incorporates a trainable spatial 
embeddings layer, which acts as generalized random effects: 
each cell receives a learned vector that can absorb whatever 
differs between cells and is not captured by the observed 
feature vector $\mathbf{X}$.

\begin{definition}[Architecture of $f_{NN}$]\label{def:architecture}
The forward pass is defined by the composition
\[
f_{NN} = f_{sp} \circ f_{out} \circ f_2 \circ f_1 \circ 
f_{cat} \circ f_{emb},
\]
where $d$ denotes the dimension of the numerical feature 
vector and $d_e = 8$ is the embedding dimension.

\textbf{Spatial index embedding.} Each cell index 
$c \in \mathcal{C}$ is mapped to a trainable dense vector 
via row lookup in the embedding matrix 
$\mathbf{E} \in \mathbb{R}^{N_c \times d_e}$:
\[
f_{emb}: \mathcal{C} \to \mathbb{R}^{d_e}, \quad 
c \mapsto \mathbf{e}_c = \mathbf{E}[c,:].
\]

\textbf{Concatenation.} The embedding and the feature vector 
are concatenated to form the joint input:
\[
f_{cat}: \mathbb{R}^{d_e} \times \mathbb{R}^{d} \to 
\mathbb{R}^{d_e+d}, \quad 
(\mathbf{e}_c, \mathbf{x}) \mapsto 
\mathbf{h}_0 = [\mathbf{e}_c \parallel \mathbf{x}].
\]

\textbf{Hidden layers.} Two fully connected layers with 
ReLU activations $\sigma(z) = \max(0,z)$:
\[
f_1: \mathbb{R}^{d_e+d} \to \mathbb{R}^{64}, \quad 
\mathbf{h}_0 \mapsto \mathbf{h}_1 = 
\sigma(\mathbf{W}_1\mathbf{h}_0 + \mathbf{b}_1),
\]
\[
f_2: \mathbb{R}^{64} \to \mathbb{R}^{32}, \quad 
\mathbf{h}_1 \mapsto \mathbf{h}_2 = 
\sigma(\mathbf{W}_2\mathbf{h}_1 + \mathbf{b}_2),
\]
where $\mathbf{W}_1 \in \mathbb{R}^{64 \times (d_e+d)}$, 
$\mathbf{b}_1 \in \mathbb{R}^{64}$, 
$\mathbf{W}_2 \in \mathbb{R}^{32 \times 64}$, 
$\mathbf{b}_2 \in \mathbb{R}^{32}$ are learnable parameters. 
Dropout regularization with rate $0.2$ is applied after each 
hidden layer during training; it is omitted from the formal 
composition above as it is inactive at inference time.

\textbf{Output projection.}
\[
f_{out}: \mathbb{R}^{32} \to \mathbb{R}^{2}, \quad 
\mathbf{h}_2 \mapsto \mathbf{z} = 
\mathbf{W}_{out}\mathbf{h}_2 + \mathbf{b}_{out},
\]
where $\mathbf{W}_{out} \in \mathbb{R}^{2 \times 32}$ and 
$\mathbf{b}_{out} \in \mathbb{R}^{2}$.

\textbf{Softplus output activation.} To guarantee strict 
positivity of $\mu$ and $\alpha$, the logit vector 
$\mathbf{z} = [z_1, z_2]^\top$ is passed through the softplus 
function with a numerical stability constant $\epsilon > 0$:
\[
f_{sp}: \mathbb{R}^{2} \to \mathbb{R}_{>0}^{2}, \quad 
\mathbf{z} \mapsto
\begin{bmatrix} \mu \\ \alpha \end{bmatrix}
= \ln\!\big(\mathbf{1} + \exp(\mathbf{z})\big) + \epsilon.
\]
The softplus function $\ln(1 + e^z)$ is chosen over the 
exponential link because it is approximately linear for 
large positive inputs, avoiding gradient saturation, while 
still guaranteeing positivity. The constant $\epsilon > 0$ 
prevents numerical underflow when $z \ll 0$.

The complete forward pass is therefore:
\[
\begin{bmatrix} \mu \\ \alpha \end{bmatrix}
= (f_{sp} \circ f_{out} \circ f_2 \circ f_1 \circ f_{cat})
\big(f_{emb}(c),\, \mathbf{x}\big).
\]
\end{definition}

\begin{remark}[Interpretation of Spatial Embeddings]
The vector $\mathbf{e}_c \in \mathbb{R}^{d_e}$ is 
functionally analogous to a random effect in a generalized 
linear mixed model (GLMM), with one key distinction: the 
latent representations are not postulated to be Gaussian, 
but are learned end-to-end. In principle this lets the network 
group cells with similar activity without a prescribed 
covariance structure; we did not check whether the learned 
vectors reflect any seismotectonic property. In the GLMM analogy, $\mathbf{E}$ 
plays the role of the random effect design matrix, but 
with a learned rather than prescribed covariance.
\end{remark}

\section{Experimental Validation}
\label{sec:experiments}

\subsection{Evaluation Protocol}
\label{sec:eval-protocol}

All models are evaluated on $\mathcal{D}_{final}$ using three 
point metrics and two probabilistic metrics.

\paragraph{Models.}
``The network'' or ``the neural NB network'' is the network of
Definition~\ref{def:architecture} trained with the NB loss (Hybrid DL NB,
\texttt{Hybrid\_DL\_Enhanced} in the code); ``the Poisson network''
(Neural Poisson) is the same network trained with the Poisson loss. ``The
linear model'' is a negative binomial generalised linear model with a log
link on the same seven standardised features (NB GLM); ``the linear model
with cell intercepts'' adds one indicator per cell, and the Poisson GLM is
the Poisson counterpart of the NB GLM. ETAS is the per-cell temporal ETAS
model~\cite{ogata1988statistical}; ``ETAS with dispersion'' takes its
forecast as the mean of a negative binomial distribution with one
dispersion parameter fitted on the training period, and ``ETAS with
within-week chains'' also simulates the aftershocks triggered inside the
forecast week.

\paragraph{Point metrics.}
\begin{align*}
\mathrm{MAE}  &= \frac{1}{N}\sum_i |y_i - \hat{\mu}_i|, \\
\mathrm{RMSE} &= \sqrt{\frac{1}{N}\sum_i (y_i - \hat{\mu}_i)^2}, \\
\mathrm{MPD}  &= \frac{2}{N}\sum_i \left[y_i \ln\!\frac{y_i}
{\hat{\mu}_i} - (y_i - \hat{\mu}_i)\right]
\end{align*}
where terms with $y_i = 0$ are replaced by $\hat{\mu}_i$, 
and $\hat{\mu}_i \ge 10^{-9}$ is enforced for numerical 
stability. MPD is the mean Poisson deviance; it is a 
mean-oriented criterion and is insensitive to the quality 
of conditional dispersion modeling. Adequacy of the 
distributional specification is assessed separately, by the 
exact tail test described below.

\paragraph{Probabilistic metrics.}
NLL denotes the negative log-likelihood of the observed count under the 
model's own predictive distribution (Poisson or NB), computed from the full 
probability mass function, $\log y!$ included; averaged over cell-weeks it is 
the logarithmic score, a proper scoring rule~\cite{czado2009predictive}, and 
because every model is scored on the same observed counts it is comparable 
across model families. CRPS denotes the discrete 
Continuous Ranked Probability Score~\cite{czado2009predictive}:
\[
\mathrm{CRPS} = \sum_{k=0}^{K_{\max}} 
\bigl(F(k) - \mathbf{1}_{y \le k}\bigr)^2,
\]
where $F(k) = \mathbb{P}_\theta(Y \le k \mid \mathbf{X})$ is 
the predicted CDF truncated at $K_{\max}$. Unlike the deviance, CRPS 
also uses the whole predictive distribution.

\paragraph{Static splits.}
A chronological 80/20 split is used: the first 80\% of 
unique weeks form the training set and the remaining 20\% 
the test set, consistent with Definition~\ref{def:omega-train}. 
For GLM models, $\alpha$ is estimated by profile likelihood 
maximization over a grid of 60 points 
$\alpha \in [10^{-3}, 10^{2}]$. For DL models, the last 10\,\% of the training weeks form a
chronological validation block for early stopping, ensuring causal
epoch selection and preventing look-ahead 
bias consistent with the remark on causality in Appendix~\ref{rem:causality}.

\paragraph{Year-by-year validation.}
\label{sec:walkforward}
For each test year $Y \in \{2018, \dots, 2023\}$, the model 
is trained on all years $< Y$ and evaluated on year $Y$. 
Formally, for fold $Y$ the training index block is
\[
\Omega_{train}^{(Y)}=\{(i,j,t)\in
\operatorname{supp}(\mathcal{D}_{final})\mid 
\mathrm{year}(t)<Y\},
\]
and the test block contains all weeks of calendar year $Y$. 
This is a separate construction and must not 
be conflated with the static 80/20 split of 
Definition~\ref{def:omega-train}; the relationship between 
the two protocols is discussed in 
the remark on their relationship in Section~\ref{rem:wf-vs-static}. Five systems are compared: the NB GLM (MLE $\alpha$ re-estimated on each fold), 
the same GLM with an intercept per cell, 
Hybrid DL NB, Neural Poisson, and ETAS per cell. DL models
in the year-by-year protocol use the last 15\,\% of the weeks of $\Omega_{train}^{(Y)}$ as a
chronological validation block for early stopping.

\paragraph{Training settings.}
Both protocols train the network of Definition~\ref{def:architecture}
(embedding dimension 8, hidden layers of 64 and 32 units with dropout 0.2,
softplus outputs with $\epsilon = 10^{-8}$) by Adam with learning rate
$10^{-3}$ on mini-batches of 64 rows, and keep the weights of the best
validation epoch. The static protocol, whose settings the blind test of
Section~\ref{sec:blind} reuses, trains for at most 100 epochs, stops after
10 epochs without an improvement of more than $10^{-6}$ in validation loss,
and halves the learning rate whenever the validation loss has not improved
for more than 3 epochs. The year-by-year protocol trains for at most 30
epochs, stops after 5 epochs without an improvement of more than $10^{-4}$,
and keeps the learning rate constant. Beyond these settings and the size of
the validation block, the two protocols differ in four implementation
details: the static protocol standardises the features on the training
weeks without the validation block and in double precision, the
year-by-year protocol on all training weeks and in single precision; they
draw the order of the mini-batches from different random-number streams;
and they present the seven inputs in a different order (the blind test
uses the order of the year-by-year protocol).

\paragraph{Random seeds.}
Every number that depends on training the network is reported
over five seeds (42, 7, 123, 2024, 999), except where a single seed is
named. What is averaged over seeds is
the score, not the forecast: the network's deviance for a year, or its
logarithmic score for a week, is the mean of the five seeds' values, and
every test is applied to that averaged score. Averaging the five
forecasts row by row and scoring the average instead is a different rule
and gives different numbers. The blind test of Section~\ref{sec:blind}
is the one place where the forecasts are averaged, and its table says
so. A count of ``how many seeds out of five'' is shown
alongside, but it does not decide a verdict: near an interval
edge that count is noise.

\paragraph{Level correction.}
A model whose total forecast is off by a constant factor looks
worse in deviance than it is. In the year-by-year protocol we
multiply each model's forecasts for year $t$ by the ratio of
observed to forecast counts over all years before $t$, so that
the correction uses no information from year $t$ itself.
Deviance is then compared over 2019--2023. Scaling all forecasts
by a constant between $10^{-9}$ and $10^{300}$ leaves the
corrected deviance unchanged to within $4 \cdot 10^{-16}$. The
comparison with the linear model without cell intercepts is reported both
before and after the correction; the other year-by-year comparisons are
reported after it.

\paragraph{Intervals.}
Moving block bootstrap~\cite{kunsch1989jackknife} over consecutive weeks: blocks of five
weeks, all 17 cells of a week resampled together, 100\,000
resamples, percentile interval at 95\,\%. Resampling whole weeks
preserves dependence within a week instead of assuming it away.

\paragraph{Exact tail test.}
For each model we compute the exact distribution of the number
of cell-weeks with $Y \ge k$ implied by its own forecasts---the
distribution of a sum of independent indicators with
probabilities $\mathbb{P}_\theta(Y_i \ge k)$, obtained by
convolution---and a two-sided $p$-value for the observed number,
equal to twice the smaller tail. We use $k = 3$ as the main tail
and report $k = 2$ alongside. This test asks the question the
paper is about: how often the model says a busy week should
happen, against how often one did. The test treats the cell-weeks as
independent given the forecasts, as every model here does; dependence
between cells within a week, which no model represents, is not part of
the null hypothesis, and the block-bootstrap interval for the observed
count is the check that allows for it. The test is close in spirit to
the number test of the CSEP
experiments~\cite{schorlemmer2007likelihood,zechar2010likelihood}, which
compares the observed number of events with the number a forecast
implies; here the count is of cell-weeks with $Y \ge k$.

\paragraph{Power.}
For each pair of models we also compute the smallest difference
in mean logarithmic score that the test could detect with 80\,\%
power at the 5\,\% level on 144 weeks: $(1.96 + 0.84)$ times the
Newey--West standard error of the mean weekly difference. Differences in
score are tested with the Diebold--Mariano test~\cite{diebold1995comparing} on weekly
differences with a Newey--West correction~\cite{newey1987simple} (lag 4); year-by-year
comparisons use a paired $t$-test over the five yearly
deviances.

\paragraph{Two protocols, never mixed.}
The two protocols differ in the nine training settings listed above.
Switching them one at a time in the static protocol (seed 42) changes its
forecasts for five of them at threshold 4.5 (the size of the validation
block, the rows and the precision of the standardisation, the
random-number stream and the order of the inputs) and for six at 3.0
(also the learning-rate schedule), so no number from one protocol is
compared with a number from the other. The year-by-year networks are not
under-trained: of the twelve trainings at threshold 4.5 on seed 42, one
reaches the 30-epoch cap and the others stop after 9 to 14 epochs.

\paragraph{Use of AI tools.}
The author used generative AI assistants, mainly Anthropic's Claude
(Claude Code), to write and test code, to re-run and cross-check the
analysis, and to draft and edit the text. All results were checked against
the code and the data, and the author takes full responsibility for the
content.

\subsection{The Tail: Poisson Forecasts and One Dispersion Parameter}
\label{sec:tail}

At threshold 4.5, six test cell-weeks, in five distinct weeks, had
$Y \ge 3$; the block-bootstrap interval for their number is 1.99--11.92.

\begin{table}[htbp]
\centering
\caption{Expected number of test cell-weeks with $Y \ge 3$ under each
model, against the six that occurred, and the two-sided $p$-value of the
exact test. Threshold $M_c = 4.5$, 2448 cell-weeks. Network rows are
reported across five seeds.}
\label{tab:tail-main}
\small
\begin{tabular}{lcc}
\hline
Model & Expected cell-weeks with $Y \ge 3$ & Exact test, $p$ \\
\hline
ETAS (Poisson)                & 0.72 & 0.00009 \\
Poisson GLM                   & 0.54 & 0.000006 \\
Neural Poisson                & rejected on 5 of 5 seeds & 0.000004--0.0048 \\
NB GLM                        & 4.84 & 0.71 \\
ETAS with dispersion          & 4.45 & 0.58 \\
Neural NB network             & 1.99--8.55 across seeds & 0.031--0.85 \\
ETAS with within-week chains  & 5.47 & 0.93 \\
\hline
\end{tabular}
\end{table}

In Table~\ref{tab:tail-main} every forecast without a dispersion
parameter is rejected; ETAS falls
short by a factor of 8.3. Every forecast that carries a dispersion
parameter is consistent with the observed count at $k = 3$, the neural
network being rejected on one seed out of five. With a Bonferroni correction
for the seven models of Table~\ref{tab:tail-main} (threshold
$0.05/7 \approx 0.007$), every verdict stays the same except the network's
rejection on one seed ($p = 0.031$), which disappears.

The picture at $k = 2$ is less clean: 11 cell-weeks occurred, with interval
4.97--18.87. ETAS with dispersion expects 16.6 ($p = 0.19$) and plain
ETAS 7.6 ($p = 0.28$). The NB GLM expects 18.8 ($p = 0.073$) and the
chain simulation 18.9 ($p = 0.065$), both close to rejection by
overshooting. The neural network overshoots more: 20.0 cell-weeks on seed 42
($p = 0.039$), and it is rejected on three seeds out of five. A
dispersion parameter estimated per observation, as the network does,
therefore does not calibrate the tail better than a single shared
parameter; at $k = 2$ it does worse.

\paragraph{Within-week chains.} The Poisson ETAS forecast ignores
aftershocks triggered by events that occur inside the forecast week.
Simulating them reproduces the observed tail at $k = 3$ (5.47 against
six), but at threshold 4.5 the simulation predicts 151.6 events over the
test period against 121 observed, where Poisson ETAS predicts 123.8; it
also overshoots at $k = 2$ (18.94, just above the interval's upper edge
of 18.87). At threshold 3.0 the simulation stays inside both intervals:
52.5 cell-weeks with $Y \ge 2$ against 51, and 18.3 with $Y \ge 3$ against 14,
predicting 334.6 events against 349 observed (Poisson ETAS: 263.7). In
mean logarithmic score at 4.5 the simulation (0.172) and ETAS with one
dispersion parameter (0.171) are almost equal, and both beat Poisson
ETAS (0.181); we did not test the difference between the first two. On
this data, one dispersion parameter delivers the tail that the full
simulation delivers, without the simulation's overshoot in level.

\paragraph{The size of the effect depends on the threshold.} At 3.0 the
observed count is 14 cell-weeks with $Y \ge 3$ (in 13 distinct weeks)
and ETAS expects 5.1
($p = 0.0006$): a shortfall by a factor of 2.8 instead of 8.3. The
Poisson network is not rejected at 3.0 at all (on seed 42 it expects 9.9, $p = 0.22$;
$p = 0.19$--0.37 across seeds). Working below the completeness threshold, as
preprint~v1 did, hides most of the effect.

\paragraph{The overdispersion test of preprint~v1.} The
likelihood-ratio test reported in~v1 was computed on the training
sample. It therefore describes the fit of the training data, not the
dispersion of forecasts, and the statement about forecasts is the exact
tail test above. We report no $p$-value for it.

\subsection{The Network Against the Linear Model and ETAS}
\label{sec:network-vs-baselines}

\paragraph{Against the linear model without cell intercepts} the network
is better on all five seeds. Raw deviance over 2018--2023 is 7.25\,\%
lower at threshold 4.5 (6.2--8.2\,\% across seeds) and 9.40\,\% lower at
3.0 (9.0--10.0\,\%). After level correction, over 2019--2023, the gain is
6.60\,\% ($p = 0.018$) and 9.65\,\% ($p = 0.032$), better in five years
out of five in both cases. In the static protocol this is also the only
pair of models whose difference in logarithmic score is large enough to
be detected on 144 weeks: on the score averaged over seeds the observed
difference is 1.07 times the detectable one (0.41 to 1.54 times on single
seeds).

\paragraph{Against the linear model with a cell intercept} the advantage
is gone. The rule was fixed before the run: the network's claim stands
only if its gain exceeds the spread of its own deviance across seeds
\emph{and} the paired test gives $p < 0.05$. At 4.5 the gain is
3.41\,\%, above the seed spread of 1.84\,\%, but $p = 0.57$ and the
network is better in only two years out of five. At 3.0 the gain is
0.25\,\% against a seed spread of 1.42\,\%, $p = 0.79$. Inside the range
of inputs seen in training, the network's gain over the plain linear
model is therefore explained by differences between cells, which a
constant per cell captures.

Even the 3.41\,\% comes from a single cell-week out of the 4420 scored
in 2019--2023. Cell
$38.00\_83.00$ had no events above 4.5 in 2010--2019, so its intercept in
the linear model diverges and its forecast hits the numerical floor of
$10^{-9}$ (157 rows). The cell's first two events, in the week of 13
April 2020, carry a deviance of 81.7 on that one row and lift that year's
mean deviance of the linear model from 0.395 to 0.487. Without that row
the network is 0.73\,\% \emph{worse} than the linear model with cell
intercepts ($p = 0.76$). At threshold 3.0 no such row exists. A shrunken
intercept would avoid the divergence; we did not fit one.

\paragraph{Against ETAS.} In the year-by-year protocol the reference
ETAS forecasts 0.905 of the observed total at 4.5. After level
correction the linear model is 6.94\,\% worse than ETAS (worse in five
years out of five, $p = 0.038$), while the network differs from ETAS by
0.11\,\% ($p = 0.95$) at 4.5 and by $-0.82$\,\% ($p = 0.26$) at 3.0,
averaged over five seeds. In the static protocol ETAS with dispersion
has a better logarithmic score than the network on all five seeds, by
0.0032 on average, significant on two of them. The verdict is taken on the
score averaged over seeds, as Section~\ref{sec:eval-protocol} prescribes:
there the observed difference of 0.0032 is still below the smallest
difference detectable on 144 weeks, 0.0045, but only by a factor of 1.4,
and the Diebold--Mariano test on the seed-averaged weekly differences gives
$p = 0.043$, on the edge of significance in favour of ETAS. A single seed
can look far calmer---on seed 42 alone the difference is 0.0004, eleven
times below what that seed's own test could detect---which is why no
verdict here rests on one seed. We therefore still cannot say that either
model is better, because 144 weeks cannot resolve a difference this small;
but the margin is thin, and it is not in the network's favour.

\subsection{Blind Test on 2024--September 2026}
\label{sec:blind}

After the analysis above was frozen---the code fixed at commit
\texttt{2db6530} of 15 September 2026, the
architecture, features and training settings written down, together with
what we expected to see---we extended the catalogue by the same query to
13 September 2026 and forecast the new period with models fitted only on
data through the end of 2023, without retraining. The new period holds
141 weeks, 2397 cell-weeks and 200 events above 4.5, all of them inside
the same 17 cells. Seven cell-weeks have $Y \ge 3$ and 18 have
$Y \ge 2$. The expectations written down before the run are kept
with the project's records and are not published.

\begin{table}[htbp]
\centering
\caption{Blind test, 2024-01-01 to 2026-09-07: 2397 cell-weeks, 200
events, seven cell-weeks with $Y \ge 3$. Level is the ratio of forecast
to observed total. Network rows are the mean of the five seeds'
forecasts, row by row.}
\label{tab:blind}
\small
\begin{tabular}{lcccc}
\hline
Model & Level & Deviance & Expected $Y \ge 3$ & Exact test, $p$ \\
\hline
ETAS (Poisson)               & 0.641 & 0.592 & 0.39  & $3\cdot10^{-7}$ \\
ETAS with dispersion         & 0.641 & 0.592 & 4.16  & 0.25 \\
Poisson GLM                  & 1.440 & 0.641 & 12.2  & $3\cdot10^{-8}$ \\
NB GLM                       & 166.9 & 28.2  & 15.2  & 0.003 \\
NB GLM with cell intercepts  & 4.10  & 1.103 & 6.18  & 0.84 \\
Neural Poisson               & 0.738 & 0.570 & 1.70  & 0.001 \\
Neural NB network            & 0.736 & 0.570 & 5.84  & 0.74 \\
ETAS with within-week chains & 0.798 & 0.590 & 5.16  & 0.52 \\
\hline
\end{tabular}
\end{table}

\paragraph{The central result replicates.} The exact tail test at
$Y \ge 3$ rejects every Poisson forecast, on data that no choice in this
paper could have seen. Of the models that carry a dispersion parameter it
rejects only the NB GLM, at $p = 0.003$, because its level blew up, which
the next paragraph explains.

\paragraph{Three things we predicted did not happen, and we report them
as they are.} First, the level of ETAS was expected between 0.7 and 1.3;
it is 0.641. Second, the network was expected to stay within 3\,\% of
the linear model with cell intercepts; instead its deviance is 0.570
against 1.103, better by 48\,\%. Third, the chain simulation was expected
to overshoot the total count by 10--40\,\%; its level is 0.798, an
undershoot, although relative to plain ETAS it still multiplies the level
by 1.25, as it does in-sample.

\paragraph{What the numbers say about why.} The blind period is
dominated by one sequence: an $M\,7.0$ event on 22 January 2024 produced
45 events above 4.5 in a single cell-week and eight in the next, and 87
of the 200 events fall in January--April 2024. The largest weekly count
in the training period was 22, and the ``events in the last 12 weeks''
feature reached 35 there against 63 now. Two of the three linear
models extrapolate beyond that range catastrophically: the NB GLM
forecasts up to 9957 events in one cell-week, and the version with cell
intercepts up to 685; the Poisson GLM reaches 22.2. Restricted to the weeks from May 2024 onward
(2091 cell-weeks, 113 events) the models differ little: deviance 0.305 for
ETAS, 0.311 for the neural network and for the chain simulation, and
0.315--0.323 for the three linear models. All of them except the chain
simulation are calibrated there, with levels between 0.99 and 1.09; the
chain simulation overshoots, with level 1.23.

The per-cell ETAS of the sequence cell had been fitted on 30 events in 14
years, the largest $M\,6.1$, and its Omori decay parameter $p$ sits at
the upper bound of the allowed range (3.0). After a week with 45 events
it forecasts 0.076 for the next week, where eight occurred. A per-cell
temporal ETAS fitted on a cell that never held a large sequence cannot
represent one.

\paragraph{What we conclude, and what we do not.} The tail result of
Section~\ref{sec:tail} is confirmed out of sample. The conclusion of
Section~\ref{sec:network-vs-baselines}---that the network's advantage is
explained by differences between cells---holds inside the range of inputs
the models were trained on, and does not hold once an input leaves that
range: there the model families differ in how they extrapolate, and that
difference, not the cell effect, decides the comparison. Whether a
per-cell intercept or an embedding is the better device is therefore the
wrong question to ask of this data; how a model behaves on an input
beyond its training range is the question that separates them.

\subsection{Evaluation Traps and Their Sizes}
\label{sec:traps}

Table~\ref{tab:traps} collects the evaluation traps this case study ran
into, each with the size it had here. Two of them---the completeness
threshold and the single seed---are large enough to decide a conclusion
on their own.

\begin{table}[htbp]
\centering
\caption{Evaluation traps met in this case study, with the size of each.}
\label{tab:traps}
\small
\begin{tabular}{p{0.28\textwidth}p{0.62\textwidth}}
\hline
Trap & What it did here \\
\hline
Completeness threshold below the catalogue's actual one &
shrinks the ETAS tail shortfall: 2.8 times at 3.0 against 8.3 at 4.5 \\
One random seed &
five of the six effects in v1's results table change sign across five seeds \\
Too few bootstrap resamples at an interval edge &
with 2000 resamples the upper edge at 3.0, $Y \ge 3$, jumps between 22.84 and 23.83, and ``inside the interval on four of five seeds'' becomes five of five at 100\,000 \\
Tail stratum defined on the test data &
v1's ``12.5\,\% lower CRPS for $Y \ge 5$'' rests on six cell-weeks, four of them in one cell and three in consecutive weeks \\
Significance carried by one cell &
v1's significant Moran test holds only with cell $38.00\_83.00$ in place; without it $p \approx 0.46$--0.49 in v1, and $p = 0.76$ and 0.80 for the two networks on the present forecasts at 3.0; at 4.5, $p = 0.046$ and 0.009 with that cell and 0.51 and 0.25 without it \\
A conclusion carried by one row &
the network's remaining 3.41\,\% gain at 4.5 comes from a single cell-week; without it the sign flips \\
An exponential link extrapolating past the training range &
after an $M\,7.0$ sequence the NB GLM forecasts up to 9957 events in a cell-week, where the largest count ever observed is 45 and the largest in training was 22 \\
An aggregate score that hides a division of labour &
over 144 weeks the network and ETAS with dispersion differ by $+0.003$ in logarithmic score at 4.5---apparently nothing. Split by the week's activity, the network is worse on quiet weeks ($+0.0089$, $p < 0.001$) and better on active ones; Spearman correlation between activity and the difference is $-0.294$ ($p = 0.0004$), and $-0.446$ at threshold 3.0 \\
Reference model with a biased level &
a constant factor in the reference's total makes it look worse in deviance than it is, so the year-by-year comparisons are reported after level correction, the first of them also before it \\
Reference fitted on a different catalogue than the target &
our first ETAS baseline was fitted on the unfiltered catalogue while the target counted only events above 4.5; the corrected level is 0.905 \\
Shared random number generator across models &
with one generator for a whole run, the order in which models are trained moves their deviance; each model is now seeded separately \\
\hline
\end{tabular}
\end{table}

Several of these lead to the same place: a conclusion carried by one cell
or by a handful of cell-weeks. With 17 cells and a few busy weeks, that
is the default risk here, not an exception.

\section{Related Work}
\label{sec:related}

\paragraph{Classical seismology: ETAS and temporal point 
processes.}
The reference model in operational seismic forecasting 
is ETAS (Epidemic Type Aftershock Sequence, 
\cite{ogata1988statistical}), which describes the 
conditional event intensity as a sum of background 
activity and a temporal superposition of Omori--Utsu 
aftershock contributions from preceding events. Its
analytical properties, including the sub- and supercritical
regimes of the aftershock cascade, are studied
in~\cite{helmstetter2002subcritical}. The Collaboratory for the
Study of Earthquake Predictability (CSEP) tests probabilistic
forecasts prospectively under common
rules~\cite{schorlemmer2007likelihood,zechar2010likelihood};
space-time ETAS forecasts run in it, for example for
Japan~\cite{zhuang2011multidimensional}. The present work includes per-cell temporal 
ETAS as a direct comparative baseline, evaluated under 
the same protocols as the neural models 
(Section~\ref{sec:eval-protocol}), and---in the form with one 
dispersion parameter---as the model that passes the tail test 
this paper turns on.

\paragraph{Count regression and NB-GLMM.}
The Negative Binomial distribution as a model for 
overdispersion in count regression is systematically 
treated in~\cite{cameron2013regression}; MLE estimation 
for NB-GLM is developed in~\cite{lawless1987negative}. 
Generalized linear mixed models (GLMM) accommodate 
multiple sources of random effects analogous to our 
spatial embeddings, but require explicit specification 
of the covariance structure. The present work replaces 
this requirement with end-to-end learning through the 
embedding layer, as formalized in 
Definition~\ref{def:architecture} and the remark that follows it. Preprint~v1 claimed a further distinction from the standard 
NB-GLM: per-observation estimates of~$\alpha$ instead of a 
single global dispersion parameter. That claim is withdrawn here 
(Section~\ref{sec:corrections}); on this data one shared 
parameter calibrates the tail at least as well 
(Section~\ref{sec:tail}).

\paragraph{Neural point processes.}
The Neural Point Processes 
literature~\cite{du2016recurrent, mei2017neural} 
generalizes the ETAS formalism to arbitrary conditional 
intensities parametrized by neural networks. 
\cite{shchur2019intensity} discuss the trade-offs of
parametrising a conditional intensity, which has to stay
positive, and model the inter-event times directly instead;
the present work keeps a positive output by
softplus$+\epsilon$, as described in 
Definition~\ref{def:architecture} with the gradient 
stability rationale given therein. The present work 
differs from this stream in that we operate on 
aggregated weekly counts rather than event times, 
which simplifies training and interpretation but 
sacrifices sub-weekly temporal structure.

\paragraph{Deep learning in seismology.}
\cite{devries2018deep} proposed a neural method for 
aftershock prediction using the stress tensor matrix 
as input. \cite{mignan2019one} reproduced this result 
with a single-layer network and showed that linear 
models are often competitive with deep architectures 
--- an important cautionary argument against 
over-engineering in seismological ML, consistent with 
our own finding that the network's advantage over a linear 
baseline disappears once the baseline is given a per-cell 
intercept (Section~\ref{sec:network-vs-baselines}). For Central
Asia, deep networks have been trained on the USGS catalogue to classify
and regress earthquake magnitudes in Kazakhstan~\cite{nurtas2025analyzing}
and to forecast strong events with a hybrid CNN--LSTM
model~\cite{nurtas2025lstm}. On a weekly $1^\circ$ grid over the same
region, \cite{nurtas2026macro} find that six different architectures
reach nearly the same skill, and conclude that calibrated baselines,
which encode each cell's historical event rate, and structured features
matter more than the architecture; our finding with a per-cell intercept
points the same way. In contrast to
these works, we address probabilistic forecasting of
a count process (weekly event counts per cell) rather
than aftershock coordinate prediction or
classification.

The model examined here sits between the two: the 
dispersion model is the statistically principled one, 
formally grounded in Theorem~\ref{thm:overdispersion} and 
Proposition~\ref{prop:nb-mixture}, while the mean is a 
neural approximation that needs no explicit physical 
triggering model. What this paper tests is whether that 
combination buys anything on a sparse catalogue, measured 
against baselines given the same information.

\section{Discussion}
\label{sec:discussion}

\paragraph{What survives.} Two statements survive a careful evaluation
of this case. Poisson forecasts of weekly counts, the temporal ETAS
model included, understate the probability of a busy week by a large
factor on this catalogue, and a single dispersion parameter is enough to
remove that shortfall: the exact tail test rejects every Poisson
forecast at $Y \ge 3$ and rejects none of the dispersion models. The
same holds out of sample, on a period downloaded after the analysis was
frozen, with one exception: the NB GLM, whose level blew up there. This is the known consequence of a branching
process~\cite{kagan2010statistical}; what this paper adds is its size
here, and the finding that a dispersion parameter estimated per
observation buys nothing over one shared parameter---at $k = 2$ it is
worse.

\paragraph{What does not survive.} The claim of preprint~v1 that the
network beats the linear baseline does not survive a baseline that is
given a per-cell intercept. Inside the training range the remaining gain
is not significant, and it rests on a single cell-week. The comparison
against ETAS is inconclusive in the other direction, and only just: the
observed difference in logarithmic score is 1.4 times smaller than what
144 weeks can detect, with ETAS ahead. Saying ``no difference was
found, and the data are too few to find one'' is the honest reading;
saying ``as good as ETAS'' would not be.

\paragraph{What the blind test changed.} Out of sample, the model
families were separated not by the per-cell term but by how they behave
on an input beyond the training range. After an $M\,7.0$ sequence the
NB GLM forecast up to 9957 events in a cell-week and its version with cell
intercepts up to 685, while the neural NB network, averaged over five
seeds as in Table~\ref{tab:blind}, stayed at or below 2.1 and the
per-cell ETAS, whose
Omori decay parameter had been driven to the edge of its allowed range,
ignored the sequence altogether. On the quiet part of the same period
every model is within 6\,\% of the others. For sparse counts this is a
more useful axis of comparison than the mean score.

\paragraph{Why the traps matter beyond this paper.} The sizes in
Table~\ref{tab:traps} are examples, not constants. What they show is the
order of magnitude of the risk: on a panel of 17 cells with a handful of
busy weeks, an evaluation choice moves a comparison by as much as the
models differ from each other. The remedy is not a better model but a
protocol whose decisions are written down before the run, which is what
we did here, and which is why the three predictions that failed are
reported rather than explained away.

\subsection{Limitations}
\label{sec:limitations}

\begin{enumerate}
\item \textbf{One region, one catalogue.} The measured sizes describe
this case. They are examples, not general constants.

\item \textbf{Model choices were made after seeing 2021--2023.} The
architecture and the features were chosen when the test years had already
been looked at. Year-by-year validation protects against tuning
parameters on the test years; it does not protect against these choices.
The blind test of Section~\ref{sec:blind} does, but only for the models
as they stood when it was frozen.

\item \textbf{ETAS here is temporal and per cell.} Where a cell has too
few events for the fit (fewer than five before the cut-off) or the fit
fails, ETAS falls back to a constant rate: at threshold 4.5 in 4 of the
17 cells in the static protocol, in 2 to 5 cells in each year-by-year fold
and in 2 cells in the blind test. A cell with no event at all before the
cut-off, which is cell $38.00\_83.00$ in the 2018--2020 folds, gets a fixed
0.1 events per week rather than a rate estimated from the data. Neighbouring cells do
not excite each other, and one cell spans about 330~km from north to south and 250~km from
east to west. A $3^\circ$ cell
is coarse for discussing clustering, and a finer grid does not fix it on
this catalogue: halving the cell to $1.5^\circ$ raises the share of empty
cell-weeks from 90.9\,\% to 97.0\,\% at threshold 3.0, and at $1^\circ$
only two cells out of 105 accumulate 100 events in 14 years, while at
$0.5^\circ$ none do. At the completeness threshold of 4.5 a single cell
reaches 100 events at $3^\circ$, and none does at $1^\circ$ or finer. The
region does not hold enough events to be both fine-grained and
estimable.

\item \textbf{Short test period, low power.} 144 weeks in the static
protocol and five scored years in the year-by-year protocol. Several
comparisons above are inconclusive for this reason, and we say so where
it applies rather than reporting the sign of a difference the data
cannot resolve.

\item \textbf{Dependence between cells is handled only by resampling
whole weeks.} The bootstrap resamples entire weeks, all 17 cells
together, so dependence within a week is preserved rather than assumed
away. The measured average pairwise correlation between cells is small
but not zero: for counts over the whole panel of 719 weeks, $+0.003$ at
threshold 3.0 and $+0.008$ at 4.5; for the ETAS residuals $y - \lambda$
over the 144 test weeks, $+0.011$ and $+0.046$. Single pairs go higher.
Among the cells that recorded at least one event in the test period, the
largest residual correlation is $+0.55$ at threshold 4.5, between
$38.00\_71.00$ and $41.00\_77.00$---and those two are not adjacent: they are
one grid step apart in latitude and two in longitude. Every pair above that, up to $+0.97$,
is a pair of cells with no events at all in the test period, where the
residual is minus an almost constant forecast; such a correlation
describes the forecast, not triggering. Moran's test on the networks'
residuals (seed 42) shows the same pattern at the main threshold as at
3.0: significant with all cells ($p = 0.046$ and 0.009 for the NB and
Poisson networks) and not without cell $38.00\_83.00$ ($p = 0.51$ and
0.25). No model in this paper has a term
for cross-cell triggering.

\item \textbf{The linear model with cell intercepts is a plain
maximum-likelihood fit.} A cell with no past events gets a diverging
intercept, which is exactly the row that carries the network's remaining
gain. A shrunken intercept would avoid this; we did not fit one.
\end{enumerate}

\section{Corrections Relative to Preprint v1}
\label{sec:corrections}

This version replaces preprint~v1. The following statements of~v1 are
withdrawn or corrected. Numbers that describe v1 itself---8.6\,\%,
12.5\,\%, the Spearman correlations $-0.708$ and $-0.934$, the seed of
v1's figure, $I = 0.19917$, $p \approx 0.46$--0.49, and the count of
v1's effects that change sign across seeds---are quoted from v1 or from
reruns of its configuration and are not produced by the pipeline released
with this version.

\begin{enumerate}
\item \textbf{Completeness.} v1 used $M \ge 3.0$, below the completeness
magnitude of 4.5 that v1 itself estimated; 55.1\,\% of its events lie
below 4.5.

\item \textbf{Code defects.} The Poisson negative log-likelihood omitted
$\log(y!)$; the NB GLM row of v1's Table~1 is not reproduced by the
code; and the year-by-year validation split its validation set by cells
instead of by time. The first and the third are fixed, and each is now held by a test
(\texttt{tests/test\_distributions.py} and
\texttt{tests/test\_walkforward\_split.py}); the NB GLM row goes with
v1's Table~1, which is removed (item~3).

\item \textbf{Table~1 of v1} (static protocol, one seed) is removed: five
of its six effects change sign across five seeds.

\item \textbf{``8.6\,\% lower deviance than the NB GLM''} in the
abstract of~v1 is replaced by Section~\ref{sec:network-vs-baselines}:
the gain over the plain linear model is 6.60\,\% at threshold 4.5 after
level correction, and it disappears against a linear model with cell
intercepts.

\item \textbf{``12.5\,\% lower tail CRPS''} rested on six cell-weeks from
three episodes, four of them in one cell. Withdrawn, and replaced by the
exact tail test of Section~\ref{sec:tail}.

\item \textbf{The map of the dispersion parameter by cell} is withdrawn.
The per-cell $\alpha$ of v1 ranked cells in the opposite order to their
clustering (Spearman $-0.708$) and to their event counts ($-0.934$).
The figure and the summary of $\alpha$ in v1 describe seed 999, not the
seed stated in the text.

\item \textbf{Moran's test.} Moran's~$I$~\NoHyper\cite{cliff1981spatial}\endNoHyper{} was used in v1 to
test the conditional independence of cells. Its significance there rested on a single
cell: without cell $38.00\_83.00$, $p \approx 0.46$--0.49 in v1, and
$p = 0.76$ (NB network) and 0.80 (Poisson network) on the present
forecasts at threshold 3.0. For the
Poisson network, v1 reported $I = 0.19917$ computed with an untrained
dispersion value of 0.786; the corrected value is $I = 0.19535$,
$p \approx 0.048$.

\item \textbf{The likelihood-ratio test for overdispersion}, with the
boundary correction of the null distribution~\NoHyper\cite{self1987asymptotic}\endNoHyper, was computed
on the training sample, so it describes the fit of the training data and
not the dispersion of forecasts. The forecast-side statement is now the
tail test of Section~\ref{sec:tail}, and no $p$-value is carried over.

\item \textbf{An editing note left in the text} of v1's tail section is
removed, together with the incorrect sentence around it.
\end{enumerate}

\section{Conclusion}
\label{sec:conclusion}

We revisited a neural negative binomial forecast of weekly earthquake
counts on a sparse panel and asked which of its conclusions survive a
careful evaluation.

What survives is a distributional statement: Poisson forecasts of weekly
counts, ETAS included, understate the probability of a busy week---by a
factor of 8.3 on this catalogue---and one dispersion parameter removes
the shortfall, while a dispersion parameter per observation does not
improve on it. The statement replicates on a period that no choice in
this work could have seen, with models fitted only on earlier data.

What does not survive is the comparative claim. Once the linear baseline
is given a per-cell intercept, the network's advantage inside the
training range is not significant and rests on a single cell-week;
against ETAS the test period is too short to resolve the difference at
all. Out of sample the families are separated by how they extrapolate
beyond the training range, not by the per-cell term.

The methodological contribution is the set of measured sizes in
Table~\ref{tab:traps}: on a panel of 17 cells, evaluation choices such
as a threshold below completeness, a single seed, or too few bootstrap
resamples move a comparison by as much as the models differ from one
another. We report the three predictions that failed in the blind test, because
a protocol that only reports its successes measures nothing.

\section*{Data and Code Availability}

Code and processed data are available at
\url{https://github.com/Al1mkaYandere/seismic-probabilistic-modeling};
the version described here is commit \texttt{bbc1dbd} of 27 September 2026.
The repository holds the USGS catalogue the analysis is built on (2 January
2010 to 27 December 2023), the catalogue of the blind period (1 January 2024
to 13 September 2026, downloaded after the code was fixed) and the processed
panels at both thresholds. The pipeline runs offline from these catalogues:
\texttt{SPM\_OFFLINE=1 python main.py} at threshold 3.0, and at 4.5 the same
command with \texttt{SPM\_M\_C=4.5}, \texttt{SPM\_PIN\_AXIS=1} and
\texttt{SPM\_PANEL\_FILE=spatiotemporal\_grid\_mc45.csv}. Each of the two
commands rebuilds its panel and all of its tables byte for byte; the tables
hold the static-protocol numbers of Section~\ref{sec:tail} and the
year-by-year forecasts behind Section~\ref{sec:network-vs-baselines}.
\texttt{SPM\_SEED} selects the training seed, and the ranges over seeds in
Table~\ref{tab:tail-main} are the minimum and maximum over seeds 42, 7, 123,
2024 and 999. \texttt{python blind/run\_blind.py} refits every model on data
through 2023, forecasts the new period without retraining and reproduces
the numbers of Table~\ref{tab:blind}. The remaining numbers---the averages
over seeds and their paired tests, the analysis without the single row of
Section~\ref{sec:network-vs-baselines}, the split by weekly activity, the
correlations between cells and the grid-size counts of
Section~\ref{sec:limitations}---are computed from these tables by short
analysis scripts that are not part of the repository.

One further check is not released: a comparison of 2744 numbers and the
checksums of 15 files against the values v1 was written from, which carries
v1's own outputs.

\appendix

\section{Formal Setting and Derivations}
\label{app:formal}

This appendix collects the formal construction behind
Section~\ref{sec:methodology} and the proofs of the results stated there.

\subsection{Formalization of the Probability Space}

\begin{definition}[Local Phase Space]\label{def:phase-space}
The space of admissible states of a node is:
\[
S = \mathbb{N}_0 \times \mathbb{R}_{\ge 0} \times \mathbb{R} 
\times \mathbb{R},
\]
equipped with the Borel $\sigma$-algebra
\[
\mathcal{B}(S) = 2^{\mathbb{N}_0} \otimes 
\mathcal{B}(\mathbb{R}_{\ge 0}) \otimes \mathcal{B}(\mathbb{R}) 
\otimes \mathcal{B}(\mathbb{R}),
\]
where $2^{\mathbb{N}_0}$ denotes the power set of the countable 
set $\mathbb{N}_0$. The four components correspond, respectively, 
to the event count $Y \in \mathbb{N}_0$, released seismic energy 
$E \in \mathbb{R}_{\ge 0}$, maximum magnitude, and minimum 
magnitude.
\end{definition}

\begin{definition}[Global Probability Space]\label{def:prob-space}
The sample space is defined as the function space 
$\omega: \Omega_{active} \to S$:
\[
\Omega_{full} := S^{\Omega_{active}} \cong 
\prod_{(i,j,\tau) \in \Omega_{active}} S.
\]
The cylindrical $\sigma$-algebra is:
\[
\mathcal{F} = \bigotimes_{(i,j,\tau) \in \Omega_{active}} 
\mathcal{B}(S).
\]
A probability measure $\mathbb{P}$ is postulated on 
$(\Omega_{full}, \mathcal{F})$, inducing the joint distribution 
of seismic events across all active nodes and time steps.
\end{definition}

\subsubsection{Process Dynamics, Filtration, and Measurability}

\begin{definition}[Canonical Coordinate Maps]\label{def:coordinates}
For each index $(i,j,t) \in \Omega_{active}$, the canonical 
coordinate map $s_{i,j}^{(t)}: \Omega_{full} \to S$ is defined by
\[
s_{i,j}^{(t)}(\omega) = \omega(i,j,t).
\]
The target variable $Y_{i,j}^{(t)}: \Omega_{full} \to \mathbb{N}_0$ 
is defined via the first canonical projection 
$\pi_1: S \to \mathbb{N}_0$:
\[
Y_{i,j}^{(t)} = \pi_1 \circ s_{i,j}^{(t)}.
\]
\end{definition}

\begin{proposition}[Measurability of Canonical 
Coordinates]\label{prop:coordinate-measurability}
For any $(i,j,t) \in \Omega_{active}$, the map $s_{i,j}^{(t)}$ is 
$(\mathcal{F}, \mathcal{B}(S))$-measurable, that is,
\[
\forall B \in \mathcal{B}(S):
\quad \bigl(s_{i,j}^{(t)}\bigr)^{-1}(B) \in \mathcal{F}.
\]
Furthermore, the projection $\pi_1: S \to \mathbb{N}_0$ is 
$(\mathcal{B}(S), 2^{\mathbb{N}_0})$-measurable. Consequently, 
$Y_{i,j}^{(t)} = \pi_1 \circ s_{i,j}^{(t)}$ is an 
$\mathcal{F}$-measurable random variable taking values in 
$\mathbb{N}_0$.
\end{proposition}

\begin{proof}
Consider an arbitrary finite index set
\[
J=\{\xi_1,\dots,\xi_n\} \subset \Omega_{active},
\qquad
\xi_k=(i_k,j_k,\tau_k).
\]
Define the corresponding finite-dimensional coordinate projection
\[
\pi_J:\Omega_{full}\to S^n,
\qquad
\pi_J(\omega)=\bigl(\omega(\xi_1),\dots,\omega(\xi_n)\bigr).
\]
For any set $B_n \in \bigotimes_{k=1}^{n}\mathcal{B}(S)$, its 
preimage
\[
\pi_J^{-1}(B_n)
=\{ \omega \in \Omega_{full}\mid \pi_J(\omega)\in B_n\}
\]
is called a cylindrical set, depending only on the coordinates 
indexed by $J$. Introduce the class of all such cylinders for 
fixed $J$:
\[
\mathcal{C}_J
=\left\{ \pi_J^{-1}(B_n)\mid
B_n \in \bigotimes_{k=1}^{n}\mathcal{B}(S)\right\}.
\]
Taking the union over all finite index sets yields
\[
\mathcal{C}
= \bigcup_{\substack{J\subset \Omega_{active}\\ |J|<\infty}}
\mathcal{C}_J.
\]
The cylindrical $\sigma$-algebra on the product $S^{\Omega_{active}}$ 
is by definition the $\sigma$-algebra generated by all 
finite-dimensional cylinders:
\[
\mathcal{F}=\sigma(\mathcal{C}).
\]
Hence, by the property of generated $\sigma$-algebras,
\[
\mathcal{C}\subseteq \mathcal{F}.
\]
Now fix $(i,j,t)\in \Omega_{active}$ and an arbitrary 
$B\in\mathcal{B}(S)$. By definition of the coordinate map,
\[
\bigl(s_{i,j}^{(t)}\bigr)^{-1}(B)
=\{ \omega\in\Omega_{full}\mid \omega(i,j,t)\in B\}.
\]
The right-hand side is a cylindrical set corresponding to the 
singleton index set $J=\{(i,j,t)\}$. Indeed, for $|J|=1$ we have 
$S^1=S$ and
\[
\bigl(s_{i,j}^{(t)}\bigr)^{-1}(B)=\pi_J^{-1}(B).
\]
Since $B\in\mathcal{B}(S)=\bigotimes_{k=1}^{1}\mathcal{B}(S)$,
\[
\bigl(s_{i,j}^{(t)}\bigr)^{-1}(B)\in \mathcal{C}_J
\subseteq \mathcal{C}\subseteq \mathcal{F}.
\]
As $B\in\mathcal{B}(S)$ was arbitrary, $s_{i,j}^{(t)}$ is 
$(\mathcal{F},\mathcal{B}(S))$-measurable.

It remains to verify measurability of $\pi_1$. Let 
$A\in 2^{\mathbb{N}_0}$ be arbitrary. By definition of the preimage,
\[
\begin{aligned}
\pi_1^{-1}(A)
&= \{s\in S\mid \pi_1(s)\in A\} \\
&= \{(r_1,r_2,r_3,r_4)\in
\mathbb{N}_0\times\mathbb{R}_{\ge 0}\times\mathbb{R}\times\mathbb{R}
\mid r_1\in A\} \\
&= A \times \mathbb{R}_{\ge 0}\times \mathbb{R}\times \mathbb{R}.
\end{aligned}
\]
Since $A\in 2^{\mathbb{N}_0}$, $\mathbb{R}_{\ge 0}\in
\mathcal{B}(\mathbb{R}_{\ge 0})$, and 
$\mathbb{R}\in\mathcal{B}(\mathbb{R})$, we obtain
\[
A \times \mathbb{R}_{\ge 0}\times \mathbb{R}\times \mathbb{R}
\in
2^{\mathbb{N}_0}\otimes\mathcal{B}(\mathbb{R}_{\ge 0})\otimes
\mathcal{B}(\mathbb{R})\otimes\mathcal{B}(\mathbb{R})
= \mathcal{B}(S).
\]
Hence $\pi_1^{-1}(A)\in\mathcal{B}(S)$ for all 
$A\in 2^{\mathbb{N}_0}$, so $\pi_1$ is 
$(\mathcal{B}(S), 2^{\mathbb{N}_0})$-measurable. Finally, the 
composition of measurable maps is measurable, so
\[
Y_{i,j}^{(t)}=\pi_1\circ s_{i,j}^{(t)}
\]
is an $\mathcal{F}$-measurable random variable taking values in 
$\mathbb{N}_0$.
\end{proof}

To ensure strict causality of the forecast, we formalize the 
information structure of the process over time: the filtration 
separates the history available to the observer from the 
unobservable future.

\begin{definition}[Natural Filtration]\label{def:filtration}
The filtration $\mathbb{F} = \{ \mathcal{F}_t \}_{t \in 
\mathcal{T}_{all}}$ is defined as
\[
\mathcal{F}_t = \sigma\big( \{ s_{k,m}^{(\tau)} \mid 
(k,m) \in \mathcal{S}_{active}, \; \tau \le t \} \big), 
\quad \mathcal{F}_{t-1} \subseteq \mathcal{F}_t \subseteq 
\mathcal{F}.
\]
The process $\{s_{i,j}^{(t)}\}$ is adapted to $\mathbb{F}$: 
$\sigma(s_{i,j}^{(t)}) \subseteq \mathcal{F}_t$.
\end{definition}

The feature vector $\mathbf{X}_{i,j}^{(t)}: \Omega_{full} \to 
\mathbb{R}^7$ is defined as $\mathbf{X}_{i,j}^{(t)} = 
\Phi(s_{i,j}^{(t-1)}, \dots, s_{i,j}^{(t-W)})$. By the 
measurability of compositions of measurable functions:
\[
\sigma(\mathbf{X}_{i,j}^{(t)}) \subseteq \mathcal{F}_{t-1}.
\]
At the same time, $\sigma(Y_{i,j}^{(t)}) \subseteq \mathcal{F}_t$, 
but $\sigma(Y_{i,j}^{(t)}) \not\subseteq \mathcal{F}_{t-1}$, 
formalizing the indeterminacy of the future given the observed 
history.

Postulating sufficiency of $\mathbf{X}_{i,j}^{(t)}$ for 
$\mathcal{F}_{t-1}$ with respect to $Y_{i,j}^{(t)}$ (Markov 
assumption), we obtain the conditional independence:
\[
Y_{i,j}^{(t)} \perp \mathcal{F}_{t-1} \mid \mathbf{X}_{i,j}^{(t)}.
\]
The problem reduces to approximating the conditional likelihood 
$\mathbb{P}(Y_{i,j}^{(t)} \mid \mathbf{X}_{i,j}^{(t)})$ with 
parameters $\theta = f_{NN}(\mathbf{X}_{i,j}^{(t)})$.

\begin{remark}[Markov Truncation and Finite Memory]
The Markov assumption, combined with the feature functional 
$\Phi$ of Definition~\ref{def:features}, reduces the 
infinite past $\mathcal{F}_{t-1}$ to a finite window of 
$W_{\max} = 12$ weeks. Formally, this means we replace the 
condition on the full sigma-algebra $\mathcal{F}_{t-1}$ with 
the condition on the finite-dimensional sub-sigma-algebra
\[
\mathcal{F}_{t-1}^{W} := \sigma\big(\{ s_{i,j}^{(\tau)} \mid 
\tau \in \{t - W_{\max}, \dots, t-1\} \}\big) 
\subseteq \mathcal{F}_{t-1}.
\]
This truncation is an assumption, not a result: aftershock rates
decay as a power law in time (the Omori--Utsu law), so events more
than 12 weeks back still carry some information, which the features
keep only through $\phi_7$.
\end{remark}

\begin{remark}[Causality]\label{rem:causality}
The inclusion $\sigma(\mathbf{X}_{i,j}^{(t)}) \subseteq 
\mathcal{F}_{t-1}$ provides a formal guarantee against 
look-ahead bias: the model operates exclusively on information 
available to the observer at the start of the forecast 
week~$t$. In the context of operational seismological 
forecasting, this means that no event from week~$t$ 
participates in the construction of the feature vector for 
that same week.
\end{remark}

\subsection{The Training Objective}
\label{app:nll}

\begin{proposition}[Equivalence of KL Divergence Minimization 
and NLL]\label{prop:nll}
Let $\mathbb{P}$ denote the true measure and $\mathbb{P}_{\theta}$ 
a parametric approximation. Then
\[
\theta^* = \arg\min_{\theta}\, \mathbb{E}_{\mathbf{X} \sim 
\mathbb{P}}\big[ D_{KL}(\mathbb{P}(Y \mid \mathbf{X}) \parallel 
\mathbb{P}_{\theta}(Y \mid \mathbf{X})) \big]
= \arg\max_{\theta}\, \mathbb{E}_{\mathbf{X}, Y \sim \mathbb{P}}
\big[ \log \mathbb{P}_{\theta}(Y \mid \mathbf{X}) \big].
\]
\end{proposition}

\begin{proof}
Expanding the KL divergence by its definition and applying the 
tower property of expectation:
\[
\mathbb{E}_{\mathbf{X}}\big[ D_{KL}(\mathbb{P} \parallel 
\mathbb{P}_{\theta}) \big]
= \mathbb{E}_{\mathbf{X}}\left[
\sum_{y} \mathbb{P}(y \mid \mathbf{X})
\log \frac{\mathbb{P}(y \mid \mathbf{X})}
{\mathbb{P}_{\theta}(y \mid \mathbf{X})}
\right].
\]
Splitting the logarithm and recognizing the conditional entropy 
$\mathcal{H}(\mathbb{P}) := -\mathbb{E}_{\mathbf{X},Y}\big[
\log \mathbb{P}(Y \mid \mathbf{X})\big]$:
\[
\mathbb{E}_{\mathbf{X}}\big[ D_{KL}(\mathbb{P} \parallel 
\mathbb{P}_{\theta}) \big]
= -\mathcal{H}(\mathbb{P})
- \mathbb{E}_{\mathbf{X}, Y \sim \mathbb{P}}\big[
\log \mathbb{P}_{\theta}(Y \mid \mathbf{X}) \big].
\]
Since $\mathcal{H}(\mathbb{P})$ does not depend on $\theta$, 
minimizing the left-hand side over $\theta$ is equivalent to 
maximizing the second term:
\[
\theta^* = \arg\max_{\theta}\, 
\mathbb{E}_{\mathbf{X}, Y \sim \mathbb{P}}\big[
\log \mathbb{P}_{\theta}(Y \mid \mathbf{X}) \big].
\]
Note that $D_{KL} \ge 0$ with equality if and only if 
$\mathbb{P}_{\theta} = \mathbb{P}$ almost surely, so the 
minimum is attained at the true distribution whenever it lies 
in the parametric family.
\end{proof}

Since $\mathbb{P}$ is unavailable, we replace the population 
expectation with its empirical counterpart. Assuming the panel is 
stationary and ergodic (the cell-weeks are not independent), the 
empirical mean converges to the expectation:
\[
\mathbb{E}_{\mathbf{X}, Y \sim \mathbb{P}}\big[
\log \mathbb{P}_{\theta}(Y \mid \mathbf{X})\big]
\;\approx\;
\frac{1}{|\Omega_{train}|}
\sum_{(i,j,t) \in \Omega_{train}}
\log \mathbb{P}_{\theta}(Y_{i,j}^{(t)} \mid 
\mathbf{X}_{i,j}^{(t)}).
\]
Since $|\Omega_{train}|$ is constant with respect to $\theta$, 
maximizing the empirical mean is equivalent to maximizing the 
sum. Taking the logarithm of the factorized joint likelihood 
(which coincides with this sum by the conditional independence 
assumption stated in Section~\ref{rem:spatial-independence}):
\[
\log \mathbb{P}_{\theta}(\mathbf{Y}_{train} \mid 
\mathbf{X}_{train})
= \log \prod_{(i,j,t) \in \Omega_{train}} 
\mathbb{P}_{\theta}(Y_{i,j}^{(t)} \mid \mathbf{X}_{i,j}^{(t)})
= \sum_{(i,j,t) \in \Omega_{train}} 
\log \mathbb{P}_{\theta}(Y_{i,j}^{(t)} \mid 
\mathbf{X}_{i,j}^{(t)}).
\]
Negating and identifying this as the empirical approximation 
to $-\mathbb{E}_{\mathbf{X},Y \sim \mathbb{P}}[\log 
\mathbb{P}_{\theta}(Y \mid \mathbf{X})]$, we obtain the 
training objective --- the negative log-likelihood (NLL):
\[
\theta^* \approx \arg\min_{\theta}\, 
\mathcal{L}_{NLL}(\theta), \qquad
\mathcal{L}_{NLL}(\theta) = -\sum_{(i,j,t) \in \Omega_{train}} 
\log \mathbb{P}_{\theta}(Y_{i,j}^{(t)} \mid 
\mathbf{X}_{i,j}^{(t)}).
\]

\subsection{Proofs}
\label{app:proofs}

\begin{proof}[Proof of Theorem~\ref{thm:overdispersion}]
We apply the law of total variance, which holds for any 
sub-$\sigma$-algebras $\sigma(\mathbf{X}) \subseteq 
\mathcal{F}^*$:
\[
\mathrm{Var}[Y \mid \mathbf{X}]
= \mathbb{E}\big[\mathrm{Var}(Y \mid \mathcal{F}^*) 
\mid \mathbf{X}\big]
+ \mathrm{Var}\big(\mathbb{E}[Y \mid \mathcal{F}^*] 
\mid \mathbf{X}\big).
\]
Since $Y \mid \mathcal{F}^* \sim \mathrm{Poisson}(\lambda^*)$ 
and $\lambda^*$ is $\mathcal{F}^*$-measurable, the Poisson 
moment identities established in 
Definition~\ref{def:poisson} give:
\[
\mathbb{E}[Y \mid \mathcal{F}^*] = \lambda^*, \qquad
\mathrm{Var}(Y \mid \mathcal{F}^*) = \lambda^*.
\]
Substituting into the law of total variance:
\[
\mathrm{Var}[Y \mid \mathbf{X}]
= \mathbb{E}[\lambda^* \mid \mathbf{X}]
+ \mathrm{Var}[\lambda^* \mid \mathbf{X}].
\]
We identify the first term as the conditional mean of $Y$. 
By the tower property, since $\sigma(\mathbf{X}) \subseteq 
\mathcal{F}^*$:
\[
\mathbb{E}[Y \mid \mathbf{X}]
= \mathbb{E}\big[\mathbb{E}[Y \mid \mathcal{F}^*] 
\mid \mathbf{X}\big]
= \mathbb{E}[\lambda^* \mid \mathbf{X}]
=: \mu(\mathbf{X}).
\]
Hence:
\[
\mathrm{Var}[Y \mid \mathbf{X}]
= \mu(\mathbf{X}) + \mathrm{Var}[\lambda^* \mid \mathbf{X}].
\]
The conditional variance is non-negative, so the inequality $\ge$ holds 
almost surely; it remains to show that $\mathrm{Var}[\lambda^* \mid 
\mathbf{X}] > 0$ on an event of positive probability. By definition, for any square-integrable 
random variable $Z$:
\[
\mathrm{Var}[Z \mid \mathbf{X}] = 0 \quad \Longleftrightarrow 
\quad Z = \mathbb{E}[Z \mid \mathbf{X}] \quad 
\mathbb{P}\text{-a.s.},
\]
which holds if and only if $Z$ is $\sigma(\mathbf{X})$-measurable 
$\mathbb{P}$-a.s. By hypothesis, $\lambda^*$ is \emph{not} 
$\sigma(\mathbf{X})$-measurable, so $\mathrm{Var}[\lambda^* \mid \mathbf{X}] > 0$ 
on an event of positive probability (otherwise 
$\lambda^* = \mathbb{E}[\lambda^* \mid \mathbf{X}]$ almost surely). Therefore:
\[
\mathrm{Var}[Y \mid \mathbf{X}]
= \mu(\mathbf{X}) + \mathrm{Var}[\lambda^* \mid \mathbf{X}]
\ge \mu(\mathbf{X})
= \mathbb{E}[Y \mid \mathbf{X}]
\quad \mathbb{P}\text{-a.s.},
\]
with strict inequality on that event.
\end{proof}

\begin{proof}[Proof of Proposition~\ref{prop:nb-mixture}]
The density of the random effect 
$\lambda^* \sim \mathrm{Gamma}(r,\beta)$ is:
\[
f_{\lambda^*}(\ell) = \frac{\beta^r}{\Gamma(r)}\,
\ell^{\,r-1} e^{-\beta \ell}, \qquad \ell > 0.
\]
By the law of total probability:
\[
\mathbb{P}(Y = y) = \int_0^\infty 
\mathbb{P}(Y = y \mid \lambda^* = \ell)\,
f_{\lambda^*}(\ell)\,d\ell.
\]
Substituting the Poisson pmf and the Gamma density:
\[
\mathbb{P}(Y = y)
= \int_0^\infty \frac{\ell^y e^{-\ell}}{y!}\cdot
\frac{\beta^r}{\Gamma(r)}\,\ell^{r-1}e^{-\beta\ell}\,d\ell
= \frac{\beta^r}{y!\,\Gamma(r)}
\int_0^\infty \ell^{\,y+r-1}\,
e^{-(1+\beta)\ell}\,d\ell.
\]
The remaining integral is a standard Gamma integral. 
Recognizing that for $a > 0$:
\[
\int_0^\infty \ell^{\,y+r-1} e^{-a\ell}\,d\ell 
= \frac{\Gamma(y+r)}{a^{y+r}},
\]
and setting $a = 1 + \beta$:
\[
\mathbb{P}(Y = y)
= \frac{\beta^r}{y!\,\Gamma(r)}\cdot
\frac{\Gamma(y+r)}{(1+\beta)^{y+r}}
= \frac{\Gamma(y+r)}{\Gamma(y+1)\,\Gamma(r)}\cdot
\frac{\beta^r}{(1+\beta)^{y+r}}.
\]
It remains to rewrite the last factor in terms of 
$p = (1+\beta)^{-1}$, so that $1-p = \beta/(1+\beta)$:
\[
\frac{\beta^r}{(1+\beta)^{y+r}}
= \left(\frac{\beta}{1+\beta}\right)^r
\cdot\left(\frac{1}{1+\beta}\right)^y
= (1-p)^r\,p^y.
\]
Substituting:
\[
\mathbb{P}(Y = y)
= \frac{\Gamma(y+r)}{\Gamma(y+1)\,\Gamma(r)}
(1-p)^r p^y,
\]
which is the Negative Binomial pmf with parameters 
$(r, p)$.
\end{proof}

\begin{proof}[Proof of Proposition~\ref{prop:nb-reparam}]
\textbf{Mean.} By the tower property and the moment 
identity $\mathbb{E}[\lambda^*] = r/\beta$ for 
$\lambda^* \sim \mathrm{Gamma}(r,\beta)$:
\[
\mu = \mathbb{E}[Y] 
= \mathbb{E}\big[\mathbb{E}[Y \mid \lambda^*]\big] 
= \mathbb{E}[\lambda^*] = \frac{r}{\beta}.
\]

\textbf{Variance.} By the law of total variance and 
the Poisson identity 
$\mathrm{Var}(Y \mid \lambda^*) = \mathbb{E}[Y \mid \lambda^*] 
= \lambda^*$:
\[
\mathrm{Var}[Y] 
= \mathbb{E}\big[\mathrm{Var}(Y \mid \lambda^*)\big] 
+ \mathrm{Var}\big(\mathbb{E}[Y \mid \lambda^*]\big)
= \mathbb{E}[\lambda^*] + \mathrm{Var}(\lambda^*).
\]
Substituting the Gamma moments 
$\mathbb{E}[\lambda^*] = r/\beta$ and 
$\mathrm{Var}(\lambda^*) = r/\beta^2$:
\[
\mathrm{Var}[Y] = \frac{r}{\beta} + \frac{r}{\beta^2}
= \mu + \frac{\mu^2}{r}.
\]
Setting $\alpha := r^{-1}$ yields 
$\mathrm{Var}[Y] = \mu + \alpha\mu^2$.

\textbf{Reparametrization of the pmf.} From the 
definitions $r = \alpha^{-1}$ and $\mu = r/\beta$, we solve 
for the canonical parameters:
\[
\beta = \frac{r}{\mu} = \frac{1}{\alpha\mu}.
\]
Substituting into $p = (1+\beta)^{-1}$ from 
Proposition~\ref{prop:nb-mixture}:
\[
p = \frac{1}{1 + \beta} = \frac{1}{1 + (\alpha\mu)^{-1}} 
= \frac{\alpha\mu}{1+\alpha\mu}, \qquad
1 - p = \frac{1}{1+\alpha\mu}.
\]
Substituting $r = \alpha^{-1}$, $p$, and $1-p$ into the 
canonical NB pmf of Proposition~\ref{prop:nb-mixture}:
\[
\mathbb{P}(Y = y \mid \mu,\alpha)
= \frac{\Gamma(y+\alpha^{-1})}{\Gamma(y+1)\,\Gamma(\alpha^{-1})}
\left(\frac{1}{1+\alpha\mu}\right)^{\alpha^{-1}}
\left(\frac{\alpha\mu}{1+\alpha\mu}\right)^y.
\]
\end{proof}

\subsection{Analytical Gradients of the NB Loss Function}

\begin{lemma}[Analytical Gradients of the Local NB 
Loss]\label{lem:gradients}
For an observation $(c,t)$ with $r_c^{(t)}=(\alpha_c^{(t)})^{-1}$ 
and $D_c^{(t)}=1+\alpha_c^{(t)}\mu_c^{(t)}$:
\[
\frac{\partial \mathcal{L}_c^{(t)}}{\partial \mu_c^{(t)}}
=\frac{\mu_c^{(t)}-Y_c^{(t)}}{\mu_c^{(t)}
\big(1+\alpha_c^{(t)}\mu_c^{(t)}\big)},
\]
\[
\frac{\partial \mathcal{L}_c^{(t)}}{\partial \alpha_c^{(t)}}
=\frac{\psi\big(Y_c^{(t)}+r_c^{(t)}\big)
-\psi\big(r_c^{(t)}\big)
-\ln D_c^{(t)}}{(\alpha_c^{(t)})^2}
+\frac{\mu_c^{(t)}-Y_c^{(t)}}{\alpha_c^{(t)}
\big(1+\alpha_c^{(t)}\mu_c^{(t)}\big)},
\]
where $\psi = (\ln\Gamma)'$ denotes the digamma function. 
Under the softplus parametrization 
$\mu_c^{(t)}=\ln(1+e^{z_{1,c}^{(t)}})$ and 
$\alpha_c^{(t)}=\ln(1+e^{z_{2,c}^{(t)}})$, the chain rule 
gives:
\[
\frac{\partial \mathcal{L}_c^{(t)}}{\partial z_{1,c}^{(t)}}
=\frac{\partial \mathcal{L}_c^{(t)}}{\partial \mu_c^{(t)}}
\,\sigma\big(z_{1,c}^{(t)}\big),
\qquad
\frac{\partial \mathcal{L}_c^{(t)}}{\partial z_{2,c}^{(t)}}
=\frac{\partial \mathcal{L}_c^{(t)}}{\partial \alpha_c^{(t)}}
\,\sigma\big(z_{2,c}^{(t)}\big),
\]
where $\sigma(z)=(1+e^{-z})^{-1}$ is the sigmoid function.
\end{lemma}

\begin{proof}
Set $\mathcal{L}_c^{(t)} = -\ell_c^{(t)}$, where 
$\ell_c^{(t)}$ is the local log-likelihood. Taking the 
logarithm of the NB pmf from 
Proposition~\ref{prop:nb-reparam}:
\[
\ell_c^{(t)}
=\ln\Gamma\!\big(Y_c^{(t)}+r_c^{(t)}\big)
-\ln\Gamma\!\big(Y_c^{(t)}+1\big)
-\ln\Gamma\!\big(r_c^{(t)}\big)
+r_c^{(t)}\ln\!\left(\frac{1}{D_c^{(t)}}\right)
+Y_c^{(t)}\ln\!\left(\frac{\alpha_c^{(t)}\mu_c^{(t)}}
{D_c^{(t)}}\right).
\]
Expanding the logarithms using 
$\ln(1/D) = -\ln D$ and 
$\ln(\alpha\mu/D) = \ln\alpha + \ln\mu - \ln D$:
\[
\ell_c^{(t)}
=\ln\Gamma\!\big(Y_c^{(t)}+r_c^{(t)}\big)
-\ln\Gamma\!\big(Y_c^{(t)}+1\big)
-\ln\Gamma\!\big(r_c^{(t)}\big)
-\big(Y_c^{(t)}+r_c^{(t)}\big)\ln D_c^{(t)}
+Y_c^{(t)}\ln\alpha_c^{(t)}
+Y_c^{(t)}\ln\mu_c^{(t)}.
\]

\emph{Gradient with respect to $\mu_c^{(t)}$.}
Note that $r_c^{(t)} = (\alpha_c^{(t)})^{-1}$ does not 
depend on $\mu_c^{(t)}$, and 
$\partial D_c^{(t)}/\partial\mu_c^{(t)} = \alpha_c^{(t)}$. 
The only $\mu$-dependent terms in $\ell_c^{(t)}$ are 
$-(Y_c^{(t)}+r_c^{(t)})\ln D_c^{(t)}$ and 
$Y_c^{(t)}\ln\mu_c^{(t)}$, giving:
\[
\frac{\partial \ell_c^{(t)}}{\partial \mu_c^{(t)}}
=-(Y_c^{(t)}+r_c^{(t)})
\frac{\alpha_c^{(t)}}{D_c^{(t)}}
+\frac{Y_c^{(t)}}{\mu_c^{(t)}}.
\]
Applying the identity $\alpha_c^{(t)} r_c^{(t)} = 
\alpha_c^{(t)} \cdot (\alpha_c^{(t)})^{-1} = 1$ to simplify 
the first term:
\[
-(Y_c^{(t)}+r_c^{(t)})\frac{\alpha_c^{(t)}}{D_c^{(t)}}
= \frac{-\alpha_c^{(t)}Y_c^{(t)} - 1}{D_c^{(t)}}.
\]
Combining over a common denominator 
$\mu_c^{(t)} D_c^{(t)}$:
\[
\frac{\partial \ell_c^{(t)}}{\partial \mu_c^{(t)}}
=\frac{(-\alpha_c^{(t)}Y_c^{(t)}-1)\mu_c^{(t)} 
+ Y_c^{(t)} D_c^{(t)}}{\mu_c^{(t)} D_c^{(t)}}.
\]
Expanding the numerator and using 
$D_c^{(t)} = 1 + \alpha_c^{(t)}\mu_c^{(t)}$:
\[
(-\alpha_c^{(t)}Y_c^{(t)}-1)\mu_c^{(t)} 
+ Y_c^{(t)}(1+\alpha_c^{(t)}\mu_c^{(t)})
= Y_c^{(t)} - \mu_c^{(t)}.
\]
Therefore:
\[
\frac{\partial \ell_c^{(t)}}{\partial \mu_c^{(t)}}
=\frac{Y_c^{(t)}-\mu_c^{(t)}}{\mu_c^{(t)}D_c^{(t)}}.
\]
Negating to obtain the loss gradient:
\[
\frac{\partial \mathcal{L}_c^{(t)}}{\partial \mu_c^{(t)}}
=\frac{\mu_c^{(t)}-Y_c^{(t)}}
{\mu_c^{(t)}\big(1+\alpha_c^{(t)}\mu_c^{(t)}\big)}.
\]

\emph{Gradient with respect to $\alpha_c^{(t)}$.}
The auxiliary derivatives are:
\[
\frac{\partial D_c^{(t)}}{\partial \alpha_c^{(t)}}
=\mu_c^{(t)},
\qquad
\frac{\partial r_c^{(t)}}{\partial \alpha_c^{(t)}}
=-\frac{1}{(\alpha_c^{(t)})^2}.
\]
We differentiate each group of terms in $\ell_c^{(t)}$ 
separately.

\textit{Gamma block.} Using $\psi(x) = \frac{d}{dx}\ln\Gamma(x)$ 
and the chain rule with 
$\partial r_c^{(t)}/\partial\alpha_c^{(t)} = 
-(\alpha_c^{(t)})^{-2}$:
\begin{align*}
\frac{\partial}{\partial \alpha_c^{(t)}}
&\Big[\ln\Gamma\!\big(Y_c^{(t)}+r_c^{(t)}\big)
-\ln\Gamma\!\big(r_c^{(t)}\big)\Big] \\
&=\Big[\psi\!\big(Y_c^{(t)}+r_c^{(t)}\big)
-\psi\!\big(r_c^{(t)}\big)\Big]
\cdot\left(-\frac{1}{(\alpha_c^{(t)})^2}\right) \\
&=\frac{\psi\!\big(r_c^{(t)}\big)
-\psi\!\big(Y_c^{(t)}+r_c^{(t)}\big)}
{(\alpha_c^{(t)})^2}.
\end{align*}

\textit{Logarithmic block.} Differentiating the remaining 
$\alpha$-dependent terms 
$-(Y_c^{(t)}+r_c^{(t)})\ln D_c^{(t)} 
+ Y_c^{(t)}\ln\alpha_c^{(t)}$:
\[
\frac{\partial}{\partial\alpha_c^{(t)}}
\Big[{-r_c^{(t)}\ln D_c^{(t)}}
-Y_c^{(t)}\ln D_c^{(t)}
+Y_c^{(t)}\ln\alpha_c^{(t)}\Big].
\]
Term by term:
\[
\frac{\partial}{\partial\alpha_c^{(t)}}
\Big[-r_c^{(t)}\ln D_c^{(t)}\Big]
=\frac{\ln D_c^{(t)}}{(\alpha_c^{(t)})^2}
-\frac{r_c^{(t)}\mu_c^{(t)}}{D_c^{(t)}}
=\frac{\ln D_c^{(t)}}{(\alpha_c^{(t)})^2}
-\frac{\mu_c^{(t)}}{\alpha_c^{(t)}D_c^{(t)}},
\]
where we used $r_c^{(t)} = (\alpha_c^{(t)})^{-1}$ in the 
last step.
\[
\frac{\partial}{\partial\alpha_c^{(t)}}
\Big[-Y_c^{(t)}\ln D_c^{(t)}\Big]
=-\frac{Y_c^{(t)}\mu_c^{(t)}}{D_c^{(t)}},
\qquad
\frac{\partial}{\partial\alpha_c^{(t)}}
\Big[Y_c^{(t)}\ln\alpha_c^{(t)}\Big]
=\frac{Y_c^{(t)}}{\alpha_c^{(t)}}.
\]
Summing the logarithmic block contributions:
\[
\frac{\ln D_c^{(t)}}{(\alpha_c^{(t)})^2}
-\frac{\mu_c^{(t)}}{\alpha_c^{(t)}D_c^{(t)}}
-\frac{Y_c^{(t)}\mu_c^{(t)}}{D_c^{(t)}}
+\frac{Y_c^{(t)}}{\alpha_c^{(t)}}
=\frac{\ln D_c^{(t)}}{(\alpha_c^{(t)})^2}
+\frac{Y_c^{(t)}-\mu_c^{(t)}}{\alpha_c^{(t)}D_c^{(t)}}.
\]

Combining the Gamma block and the logarithmic block, 
then negating:
\[
\frac{\partial \mathcal{L}_c^{(t)}}{\partial \alpha_c^{(t)}}
=\frac{\psi\!\big(Y_c^{(t)}+r_c^{(t)}\big)
-\psi\!\big(r_c^{(t)}\big)-\ln D_c^{(t)}}
{(\alpha_c^{(t)})^2}
+\frac{\mu_c^{(t)}-Y_c^{(t)}}
{\alpha_c^{(t)}\big(1+\alpha_c^{(t)}\mu_c^{(t)}\big)}.
\]

\emph{Chain rule for the pre-output logits.}
Under the softplus parametrization, the derivative of 
$\mu_c^{(t)} = \ln(1+e^{z_{1,c}^{(t)}})$ with respect to 
$z_{1,c}^{(t)}$ is:
\[
\frac{\partial \mu_c^{(t)}}{\partial z_{1,c}^{(t)}}
= \frac{e^{z_{1,c}^{(t)}}}{1+e^{z_{1,c}^{(t)}}}
= \sigma\!\big(z_{1,c}^{(t)}\big),
\]
and analogously 
$\partial\alpha_c^{(t)}/\partial z_{2,c}^{(t)} = 
\sigma(z_{2,c}^{(t)})$. 
Applying the chain rule:
\[
\frac{\partial \mathcal{L}_c^{(t)}}{\partial z_{1,c}^{(t)}}
=\frac{\partial \mathcal{L}_c^{(t)}}{\partial \mu_c^{(t)}}
\,\sigma\!\big(z_{1,c}^{(t)}\big),
\qquad
\frac{\partial \mathcal{L}_c^{(t)}}{\partial z_{2,c}^{(t)}}
=\frac{\partial \mathcal{L}_c^{(t)}}{\partial \alpha_c^{(t)}}
\,\sigma\!\big(z_{2,c}^{(t)}\big).
\]
Averaging over a mini-batch $\Omega_{batch} \subseteq 
\Omega_{train}$ yields the stochastic gradient estimate 
$\nabla_\theta \mathcal{L}_{NB}$ used in backpropagation.
\end{proof}

\end{document}